\documentclass[lettersize,journal]{IEEEtran}
\IEEEoverridecommandlockouts
\usepackage{cite}
\usepackage{amsmath,amssymb,amsfonts,amsthm}
\usepackage[lined, ruled, linesnumbered, vlined]{algorithm2e}
\usepackage{algpseudocode}
\algrenewcommand\algorithmicindent{-5em}
\usepackage{graphicx}
\usepackage{subfigure}
\usepackage{textcomp}
\usepackage{xcolor}
\usepackage{multirow}
\usepackage{multicol}
\usepackage[flushleft]{threeparttable}
\usepackage{tabularx}
\usepackage{array}
\usepackage[bookmarks=false]{hyperref}
\def\BibTeX{{\rm B\kern-.05em{\sc i\kern-.025em b}\kern-.08em
    T\kern-.1667em\lower.7ex\hbox{E}\kern-.125emX}}
\usepackage{makecell}

\newtheorem{definition}{Definition}	
\newtheorem{theorem}{\bf Theorem}		
\newtheorem{lemma}{\bf Lemma}		
\newcounter{appdx}
\newcommand{\kETAL}     {{\em et~al.}}		
\newcommand{\kEG} {{\em e.g.}}

\newcommand{\black}[1]{{\color{black} #1}}

\definecolor{myorange}{RGB}{255, 102, 0}

\newcommand{\revat}[1]{\black{#1}}
\newcommand{\revbt}[1]{\black{#1}}
\newcommand{\revct}[1]{\black{#1}}

\newcommand{\grey}[1]{{\color{black} #1}}
\definecolor{dgreen}{rgb}{0,0.655,0.149}

\newcommand{\red}[1]{{\color{black} #1}}
\newcommand{\icdcs}[1]{{\color{black} #1}}
\newcommand{\nosemic}{\renewcommand{\@endalgocfline}{\relax}}
\newcommand{\dosemic}{\renewcommand{\@endalgocfline}{\algocf@endline}}
\let\oldnl\nl
\newcommand{\nonl}{\renewcommand{\nl}{\let\nl\oldnl}}

\newcommand{\myendbox}{\hfill \qed}

\renewcommand{\S}{Section~}

\newcommand{\xjw}[1]{\black{#1}}

\IEEEaftertitletext{\vspace{-0.35in}}
\makeatletter
\renewenvironment{proof}[1][\proofname]{\par
  \pushQED{\qed}%
  \normalfont
  \topsep0pt \partopsep0pt 
  \trivlist
  \item[\hskip\labelsep
        \itshape
    #1\@addpunct{.}]\ignorespaces
}{%
  \popQED\endtrivlist\@endpefalse
  \addvspace{6pt plus 6pt} 
}
\makeatother

\begin{document}

\title{SA$^2$FE: A Secure, Anonymous, Auditable, and Fair Edge Computing Service Offloading Framework}

\author{Xiaojian~Wang,~\IEEEmembership{Graduate Student Member,~IEEE,}
        Huayue~Gu,~\IEEEmembership{Graduate Student Member,~IEEE,}
        Zhouyu~Li,~\IEEEmembership{Graduate Student Member,~IEEE,}
        Fangtong~Zhou,~\IEEEmembership{Graduate Student Member,~IEEE,}
        Ruozhou~Yu,~\IEEEmembership{Senior Member,~IEEE,}
        Dejun~Yang,~\IEEEmembership{Senior Member,~IEEE,}
        and~Guoliang~Xue,~\IEEEmembership{Fellow,~IEEE}

\thanks{%
Wang, Gu, Li, Zhou, Yu (\{xwang244, ryu5, hgu5, zli85, fzhou\}@ncsu.edu) are with North Carolina State University, Raleigh, NC 27606, USA.
Yang (\url{djyang@mines.edu}) is with Colorado School of Mines, Golden, CO 80401, USA.
Xue (xue@asu.edu) is with Arizona State University, Tempe, AZ 85287, USA.
The research of Wang, Yu, Gu, Li, Zhou was supported in part by NSF grants 2045539, 2414523 and 2433966.
The research of Xue was sponsored in part by the Army Research Laboratory and was accomplished under Cooperative Agreement Number W911NF-23-2-0225. The views and conclusions contained in this document are those of the authors and should not be interpreted as representing the official policies, either expressed or implied, of the Army Research Laboratory or the U.S. Government. The U.S. Government is authorized to reproduce and distribute reprints for Government purposes notwithstanding any copyright notation herein.

This work has been submitted to the IEEE for possible publication. Copyright may be transferred without notice, after which this version may no longer be accessible.
}
}

\maketitle


\begin{abstract}
The inclusion of pervasive computing devices in a democratized edge computing ecosystem can significantly expand the capability and coverage of near-end computing for large-scale applications.
However, offloading user tasks to heterogeneous and decentralized edge devices comes with the dual risk of both endangered user data security and privacy due to the curious base station or malicious edge servers, and unfair offloading and malicious attacks targeting edge servers from other edge servers and/or users.
\revat{Existing solutions to edge access control and offloading either rely on ``always-on'' cloud servers with reduced edge benefits or fail to protect sensitive user service information.}
To address these challenges, this paper presents SA$^2$FE, a novel framework for edge access control, offloading and accounting.
\revat{We design a rerandomizable puzzle primitive and a corresponding scheme to protect sensitive service information from eavesdroppers and ensure fair offloading decisions, while a blind token-based scheme safeguards user privacy, prevents double spending, and ensures usage accountability.}
The security of SA$^2$FE is proved under the Universal Composability framework, and its performance and scalability are demonstrated with implementation on commodity mobile devices and edge servers.

\end{abstract}




\vspace{-4mm}
\section{Introduction}
\vspace{-1mm}

\noindent As real-time computation-intensive applications such as metaverse~\cite{duan2021metaverse}, cloud gaming~\cite{meng2023enabling}, and autonomous driving~\cite{bhardwaj2022ekya} continue to grow, service providers are increasingly deploying services closer to the users~\cite{edgeComputingMarket}.
Edge computing can greatly improve user experience and reduce the cost of service providers, by achieving low latency, high reliability and backhaul communication efficiency. 
An increasing number of companies are entering the arena~\cite{edgeComputingCompany}.

Meanwhile, the rise of the Pervasive Edge Computing (PEC) paradigm~\cite{ning2020distributed}, which utilizes the computing capabilities of varied and decentralized devices as edge servers, accelerates the expansion of the edge server provider landscape, by \emph{democratizing} the edge computing ecosystem and leveraging power of the crowd.
A PEC ecosystem may involve many large and small edge providers, including but not limited to telecom companies, road-side unit operators, private infrastructure owners, and even ad hoc providers such as individuals with spare computing devices~\cite{tourani2020democratizing}.
In most cases, a telecom company provides a connection access point for edge providers and users, and usually an accompanied edge service discovery procedure for users to access the available services~\cite{dougherty2021apecs}.

However, with the expansion of the edge computing ecosystem, and especially PEC with decentralized providers, 
both edge server owners and users encounter challenges in providing and utilizing trustworthy edge computing services.
To foster the sustainable development of the edge computing market, it is imperative to design and develop technical approaches that 
can safeguard user rights, protect stakeholder interests, and maintain healthy competition.

\xjw{
In the PEC environment, a very important and indispensable part is service discovery, which is used to find available services nearby.
\grey{In traditional service discovery within Named Data Networking or Information-Centric Networking, the requesting service identities are generally exposed to surrounding devices to efficiently allocate available services to the requesting users~\cite{kaiser2014efficient}.}
Service identities, such as names or types, if descriptive or inferable, could reveal information about the nature of the data or services, potentially exposing sensitive or proprietary information to anyone who can intercept this information~\cite{wu2016privacy}.
For instance, the type of service requested by a user (such as service name or identifier) could be misused in various ways, such as profiling and identifying a user~\cite{welke2016differentiating} or inferring sensitive user attributes (e.g., inferring that a user requesting video-based visual assistance has a visual impairment~\cite{weiss2018survey}). If intercepted or accessed by malicious entities, sensitive service names could be used to perform targeted attacks, including data breaches and service disruptions.
Protecting these service types/names from unauthorized disclosure is critical to maintaining the integrity and reliability of the network. 
\grey{Some privacy-preserving service discovery approaches have been proposed to protect service request privacy~\cite{zhu2004prudentexposure}, but only within traditional discovery environments. }
However, in a PEC environment, users also have this need but lack available solutions to protect service information. Additionally, users may not want to leak their identity, preferring to remain anonymous to protect their privacy.}

\xjw{In addition,}
given the diverse stakeholders involved in service offloading, fairness in the service offloading process is paramount. 
\revct{Fairness here refers to the equitable treatment of service requests among edge devices with comparable capabilities and minimal latency differences from the user's perspective, made possible by service provider profiling and testing to ensure that multiple edge servers with similar abilities compete for task assignments.}
The assurance of fairness should not be determined by any single entity, be it the end-user, the base station overseeing offloading, or specific edge servers.
Furthermore, 
ensuring financial accountability of the users, and appropriate compensation of the providers, is critical to ensuring longevity of the market.
Unfortunately, the above security and privacy guarantees are lacking in existing works~\cite{dougherty2021apecs, zhou2022aadec}, and offloading fairness and auditability have not been addressed in an ecosystem with untrusted parties.

\revat{In this paper, we design SA$^2$FE, an innovative framework for anonymous, auditable, and fair service offloading in a PEC environment, addressing key challenges in service offloading posed by the presence of multiple competing edge server infrastructure providers.
At the core of our approach is a rerandomizable puzzle primitive, which we define and design as the foundation of our framework to enable fairness in the offloading process. 
Building on this primitive, we propose a comprehensive framework that specifies detailed interaction protocols among all participating parties, ensuring offloading fairness while protecting service type privacy.
To safeguard user identity while maintaining authorized access and accountability, we propose an anonymous token-based service request scheme. 
We design a rerandomizable puzzle-based scheme that allows offloading a service request to an eligible edge server without revealing any service-specific details to the offloading broker, or edge server-specific details to the user.  
We rigorously prove SA$^2$FE's security and demonstrate its practicality on commodity devices.}

Our contributions are summarized as follows:
\begin{itemize}
    \item We propose SA$^2$FE, 
    a secure and efficient offloading framework that preserves the privacy of user identity and requested service type, ensures fairness in edge server selection, and incorporates auditing for accountability.
    
    \item We present a novel puzzle-based offloading protocol to protect service type confidentiality while ensuring fair and randomized edge server selection. Two implementations based on bilinear map and universal re-encryption respectively are proposed to realize the puzzle scheme.

    \item We propose a token-based service access scheme that maintains user and service type confidentiality while enabling accountable token verification and claiming.
    \item We formally prove the security of SA$^2$FE under the Universal Composability (UC) framework.
    \item We implement and evaluate a prototype of SA$^2$FE on commodity mobile and edge devices. 
    The experimental results show that SA$^2$FE has low computation and communication overhead and is efficient and scalable.
\end{itemize}

\noindent
\textbf{Organization.}
\revbt{
\revbt{\S\ref{sec:related-work} reviews related work.
{\S\ref{sec:system-model} introduces the system models.}
\S\ref{sec:overview} gives an overview of SA$^2$FE.
\S\ref{sec:solution} presents the detailed design of SA$^2$FE.
\S\ref{sec:security-analysis} presents security analysis of SA$^2$FE.
\S\ref{sec:eval} shows its performance.
\S\ref{sec:conclusions} concludes this paper. 
}
}



\vspace{-2mm}
\section{Related Work}
\label{sec:related-work}

\noindent 
One related aspect of cloud and edge computing security is the access control problem.
\grey{In cloud computing, access control mainly focuses on protecting security and confidentiality of user data hosted on third-party cloud storage~\cite{xue2018combining}}
Some have studied secure and privacy-preserving data sharing through a centralized cloud~\cite{zheng2022towards,hu2020ghostor}.
The cloud provider plays a central role in facilitating access control as the single party involved. 
APECS~\cite{dougherty2021apecs} is the first distributed, multi-authority access control scheme in a dynamic pervasive edge computing ecosystem.
However, APECS mainly focuses on user data access control, neglecting anonymity and privacy preservation during service offloading, and fairness considerations. 
\grey{AADEC~\cite{zhou2022aadec} focuses on access control, prioritizing data exchange over offloading at the base station, with auditing confined to user data rather than service offloading.}

User data privacy has been studied in either cloud or edge offloading.
\grey{Li \kETAL~\cite{li2022dsos} proposed a system for solving overdetermined linear equations, ensuring privacy via permutation and validity through a detection algorithm.}
Mao \kETAL~\cite{mao2020privacy} proposed 
combining differential privacy and secure model weight aggregation to ensure privacy-preserving offloading of DNN training tasks.
Chen \kETAL~\cite{chen2013new} proposed a secure outsourcing algorithm for modular exponentiations in the one-malicious version of the two untrusted program model.

\grey{Many have studied edge offloading focusing on improving offloading performance subject to limited resources, such as resource provisioning~\cite{chen2019efficient}, task partitioning~\cite{gao2021task}, task selection~\cite{gao2019dynamic}, load balancing~\cite{park2018cooperative}, etc.}
\revat{Some recent works focus on task offloading in various edge computing scenarios using different methods, such as at intersections with game theory~\cite{zhang2025optimizing}, in satellite networks using queuing theory~\cite{jia2024deep} or game theory~\cite{chen2024multi}, in online offloading scenarios employing Deep Reinforcement Learning~\cite{wang2024ddqn}, in Aerial Mobile Edge Computing Networks through joint optimization~\cite{sun2024joint}, and using pairing theory to match services~\cite{li2025edge}.
These methods assume trust among all parties and overlook privacy concerns.}

To summarize, while existing work has addressed certain individual security concerns such as authentication, access control, user data privacy and location privacy, there lacks a comprehensive
framework for service offloading in a democratized edge computing ecosystem
that ensures offloading security, user identity and service anonymity, token accountability, and offloading fairness all at once.
Our proposed framework SA$^2$FE not only fills this gap and ensures secure offloading, but is also highly efficient, scalable, and compatible with  commodity mobile devices and edge servers.



\vspace{-3mm}
\section{Models and Problem Statement}
\label{sec:system-model}

\vspace{-1mm}
\subsection{System Model}

\noindent 
Fig.~\ref{model} shows the involved parties and their interactions.
\noindent
SA$^2$FE involves five parties: financial authority (FA), service provider (SP), base station (BS), edge server (ES) and user:
\begin{enumerate}
    \item \textbf{FA}: The FA receives payment from users, and distributes tokens for service access. It also handles reward claims from BS and ES with valid tokens as proof of service.

    \item \textbf{SP}: 
    An SP owns a service and delegates it to ESs from various registered edge server infrastructure providers, delivering edge-based services to authorized users.

    \item \textbf{BS}: 
    The base station is the broker between users and ESs. 
    It assists users within its range by discovering available ESs and offloading tasks of supported services.
    All communication between users and ESs will go through the BS.

    \item \textbf{ES}: 
    An ES provides services to users on behalf of SPs and
    may offer multiple services owned by different SPs.

    \item \textbf{User}:
    A user requests task offloading for a service that she is subscribed to, and needs to provide payment (or proof of it) to utilize an edge-offloaded service.

\end{enumerate}

\begin{figure}[tt]
    \centerline{\includegraphics[width=0.38\textwidth]{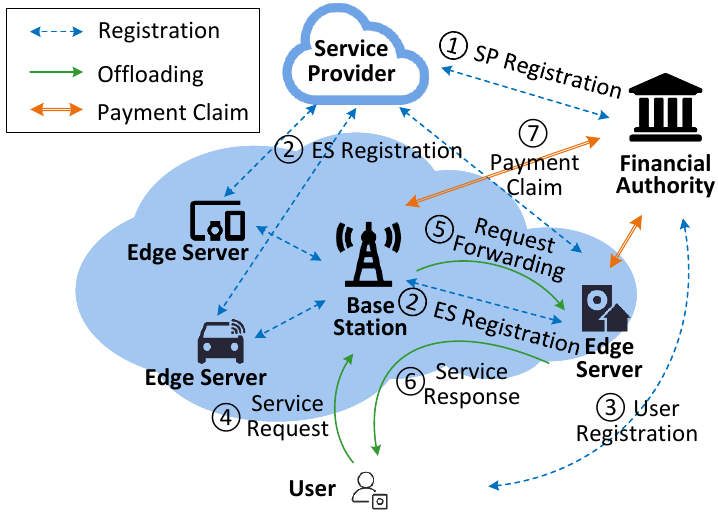}}
    \caption{
    \revat{
    SA$^2$FE workflow. 
    (1) SP registers to FA;
    (2) ES registers to SP and BS;
    (3) User gets tokens from FA;
    (4) User starts service request;
    (5) Request is forwarded to an ES; 
    (6) User gets response;
    (7) BS and ES claim tokens.
    The workflow consists of three main phases: registration (steps (1)$-$(3)), offloading (steps (4)$-$(6)), and payment claim (step (7)).
    }
    }

    \label{model}
\end{figure}

    \grey{We focus on a single BS but can be extended to multiple BSs managing different regions with shared information.}

\noindent \textbf{Interactions.}
Fig.~\ref{model} shows the interactions among these parties during an offloading process, with numbered  steps as follows:
\begin{enumerate}
\item An SP registers service-related information with the FA.
\item An ES registers with an SP to obtain the service program (\kEG~a virtual machine, container image or microservice), and service credentials to authenticate itself to users.

\item 
A user registers with the FA to deposit service pre-payments and obtain service access credentials (tokens).
\item 
When the user is within a BS which has connected ESs, the user requests offloading of her tasks from the BS.
\item 
A connected ES which is eligible to serve the request will be selected, and the BS forwards the service request to it.

\item The ES responds to the request, and starts the actual offloading process with the user.
\item The BS and ES claim pre-negotiated rewards from the FA with valid proof of service. 
\end{enumerate}

\grey{Existing work typically assumes full trust among the user, BS, and all ESs.}
\grey{For instance, the BS would neither intercept nor infer the user's service data or type and would ensure fair offloading to all ESs.}
\grey{ESs would provide services without intercepting user identity or unregistered service data.}
Users would not interfere the offloading process 
or engage in double spending for services.
\grey{In practice, these assumptions would not always hold, especially in a democratized ecosystem where neither the user nor surrounding parties can be fully trusted.}

\subsection{Threat Model}\label{threat_model}
\noindent 
We assume global parties (FA and SPs) will diligently adhere to the offloading protocol.
This is because each global party may serve many users and stakeholders, and is commonly bound by reputation to perform honestly.
In the mean time, 
local parties may deviate from the designated protocols to launch active attacks, such as a BS of a small regional Internet Service Provider (ISP), or an ES from a local provider. Similarly, a user is not  trusted to execute the protocol diligently.

A user may exhibit malicious behavior, such as expressing a preference for specific ESs, to disturb the fairness of service offloading.
Furthermore,
a user may target a specific ES to either perform reconnaissance attack in order to identify potential vulnerabilities of the ES and launch further attacks, or conduct targeted denial-of-service attacks to overwhelm target ES's resources.
A user may also try to deceive both the BS and ES by engaging in double spending, utilizing the same payment to acquire multiple offloading services, possibly from different ESs. She may also use fraudulent authorization to acquire services without making valid payments.

Meanwhile, for a user,
\grey{all other parties may possess a curiosity regarding the user's real identity and data for purposes such as data mining, targeted advertising, extracting personal information, and user tracking. }
\grey{Additionally, the user may want to hide her requested type of service from parties other than those required in purchasing and fulfilling the service, such as the BS and any ES that is not eligible to provide the service.}
Leaking the service type to untrusted parties compromises user privacy and poses security risks. 
For example, a medical offloading task exposes sensitive health information, 
while disclosure of the service type in financial transactions, location-based services, and personal preferences also poses privacy risks.
\grey{The BS and ES may also exaggerate rewards, compromising FA integrity and user payments.}

To summarize, 
We consider the following attack scenarios:
(a) A malicious user may attempt to acquire services from ESs by double spending or forging payment proofs.
(b) A malicious user may try to identify and request service from a specific ES, for instance, to disturb offloading fairness, perform reconnaissance of the ES's system, or launch denial-of-service attacks against the ES.
(c) The BS and non-eligible ESs may be curious about users' real identity, data and the service type of the request.
(d) The BS and ESs may be curious about a users' real identity, and the FA and SPs may want to link a user's identity with the time and location that she accesses a paid service.
(e) The BS/ES may exaggerate the service it provided to get extra rewards from FA.

\grey{We assume a requested service is non-identifiable except by its service type, as many services, like video analytics, share similar traffic patterns despite differing tasks.}
\grey{There are also some studies on hiding traffic patterns from eavesdroppers, such as task partitioning or traffic padding~\cite{aloufi2021edgy}.}

\vspace{-4mm}
\subsection{Problem Statement}
\noindent
Let $\mathbb{U}$, $\mathbb{S}$, $\mathbb{E}$, $\sf{B}$ be the set of users, set of services, set of ESs and the BS respectively.
We consider an offloading scenario where a user $u \in \mathbb{U}$ offloads a task of service $s \in \mathbb{S}$ to an ES $e \in \mathbb{E}$ through the BS $\sf{B}$. 
The user tries to conceal her identity from all other parties,
and keep the service type hidden from the BS and non-eligible ESs.
We require that neither the user nor the BS can ``assign'' an ES to serve a specific request; instead an eligible ES must be \emph{randomly} selected for a specific request. This both deters user reconnaissance and other malicious behaviors against a specific ES, and ensures fairness in offloading to promote community-wide sustainability and equal opportunities for all eligible ESs.
The offloading service should only be provided when the user shows proof of payment, and the reward can only be claimed when the BS/ES shows proof of service, without double spending or exaggerated claiming.

\noindent \textbf{Security goals.}
\noindent 
Our main security goals are as follows:
\begin{enumerate}
    \item 
    \textbf{Authenticated and authorized access}:
    Access to an ES-provided service is only limited to paid users of the service. 

    \item 
    \textbf{Identity privacy}:
    A user's identity is kept confidential from other parties during and after the offloading process.

    \item 
    \textbf{Service data confidentiality}:
    The service data of a user is only accessible to an eligible ES providing the service.
    
    \item 
    \textbf{Service type confidentiality}:
    The BS and non-eligible ESs have no knowledge about the requested service type.

    \item 
    \textbf{Financial accountability}:
    A user cannot get more than the paid service using invalid or double-spent payment proof.
    BS or an ES cannot claim reward without fulfilling a request, or claim multiple rewards for one service fulfillment.

     \item 
    \revct{\textbf{Offloading fairness}: ESs eligible for selection by the user are assigned an equal probability of serving user requests, ensuring fairness in service allocation.}
\end{enumerate}

\revct{We focus on scenarios where the SP has performed profiling or testing of ESs to ensure that the capabilities and latency of ESs within the pool available to the user are similar. From the user's perspective, the quality of service or quality of experience provided by these ESs in the pool is equivalent or nearly identical. An out-of-band profiling phase enables the SP to assess the servicing capabilities of ESs~\cite{do2011profiling}. 
This ensures fairness by simplifying the selection process while reducing user burden, avoiding biases caused by performance differences, and protecting ESs from being targeted by malicious users. 
Even in such scenarios, achieving fairness under our security goals is challenging, as it requires preserving service type confidentiality, ensuring efficient task allocation, and avoiding security risks or excessive overhead.
}



\section{SA$^2$FE Overview}
\label{sec:overview}

\begin{figure}[tt]
    \centerline{\includegraphics[width=0.4\textwidth]{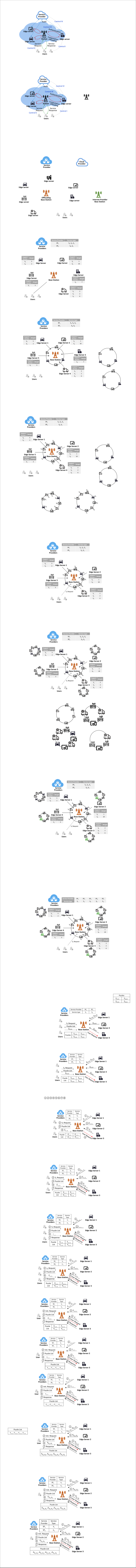}}
    \caption{
    Rerandomizable puzzle-based offloading example. 
    Suppose there are two SPs offering two types of services, $s_1$ and $s_2$, and three ESs $e_1$, $e_2$ and $e_3$ attempting to assist the SPs in delivering services. 
    $e_1$ provides both $s_1$ and $s_2$, $e_2$ provides $s_2$, and $e_3$ provides $s_1$.
    Denote the puzzle of service $s_j$ from edge server $e_i$ as $z_{e_i,s_j}$.
    Suppose a user intends to use service $s_1$.
    (1) ESs register service-related puzzles with the BS.
    (2) User initiates a request for an (unspecified) offloading service.
    (3) BS responds to user with a puzzle list.
    (4) User selects a puzzle $z_{e_3,s_1}$ and returns it to the BS.
    (5) BS forwards the user's request to selected ES $e_3$.
    (6) $e_3$ returns service response to the BS.
    (7) BS forwards the response to the user.
    }
    \label{overview}
    \vspace{0.5em}
\end{figure}

\noindent

\noindent
SA$^2$FE operates in four phases: system initialization, registration, offloading, and payment claim.
In \textbf{system initialization}, system parameters are independently initialized by all parties.
In \textbf{registration}, an ES obtains permission to host a service from an SP, and then registers service information, encrypted as \emph{cryptographic puzzles}, with a BS to serve offloading requests from local users.
A user also registers with an SP and makes payment through an FA to obtain \emph{blind tokens} for requesting the service.
In \textbf{offloading}, a user requests an offloading service through the BS, solves the service-specific puzzles and randomly picks an ES capable of providing the service.
The BS then forwards the (encrypted) request to the user-selected ES without knowing the service type.
The request contains tokens specific to the requested service, such that the BS can check for any potential double spending (again without knowing the service type), and both BS and ES can later claim service payments from the FA with the tokens.
In \textbf{payment claim}, the FA verifies the tokens submitted by the BS or ESs, and makes payment accordingly.

In SA$^2$FE, the two key building blocks are: blind tokens for service access, and cryptographic puzzles for offloading.

\noindent \textbf{Token-based service access.}
We develop a token-based service access scheme that ensures user anonymity. 
Our scheme maintains secrecy of the service type from the BS while allowing service type verification by the ESs. 
\revct{The token, specifically designed for authenticated and anonymous access to edge services, is blindly signed during user registration based on the blind signature scheme to preserve user privacy.}
A token contains a service-agnostic part for the BS, and a service-specific part for the ES signed by a service-related key.
This allows the BS to check the token for potential double spending without knowing the service type, and the ES and FA to check the service type for potential token misuse during offloading and reward claiming respectively.

\noindent \textbf{Puzzle-based offloading.}
To prevent the BS from learning a user's requested service type while still allowing forwarding the request to an eligible ES, we design a puzzle-based offloading process as shown in Fig.~\ref{overview}.
\revct{The puzzle mechanism is essential for secure and efficient offloading in PEC environments. By preventing the BS from learning service type information and ensuring that users cannot target specific ESs, puzzles  address threats like service type inference, malicious targeting, and fairness violations. Additionally, they enable a practical offloading process by eliminating the need for inefficient methods, such as random assignments or broadcasts by BS, which would otherwise increase complexity and resource consumption.}
Instead of registering plaintext service information at the BS, each ES will generate service-specific random puzzles that are indistinguishable from puzzles of other services.
The puzzles can only be solved by users with each specific service's key obtained during user registration and payment.
To ensure fair offloading, the BS sends all puzzles to the requesting user, who then solves the puzzles of its requested service, and then randomly picks one puzzle representing a random ES who can serve the service.
The puzzle contains no identifiable information about the ES, ensuring that the user cannot identify or target a specific ES during the selection (thus ensuring both fairness and protection of the ES).
Further, after every service request, the BS re-randomizes all puzzles and permutes the puzzle list to ensure that puzzles from two requests are unlinkable.

\subsection{Preliminary: Blind Signature}
\noindent SA$^2$FE makes use of a blind signature scheme as a building block, which we shall describe here for completeness.

\revct{
Blind signature~\cite{fuchsbauer2020blind} is an unlinkable digital signature scheme that allows a signer to sign a message without knowing the message content. 
The algorithms are specified as follows:

\begin{enumerate}
    \item $\text{BlindSetup}(1^\lambda) \rightarrow (PK,SK)$: Given security parameter $\lambda$, it outputs the public key $PK$ and secret key $SK$.
    
    \item $\text{BlindMsg}(PK, m, r) \rightarrow m'$: Takes $PK$, message $m$, and random number $r$ as input, outputs blinded message $m'$.
    
    \item $\text{BlindSign}(PK, SK, m') \rightarrow s'$: Takes $PK$, $SK$, and blinded message $m'$ as input, outputs signature $s'$.
    
    \item $\text{UnblindSign}(PK, s', r) \rightarrow s$: Takes $PK$,  $s'$, and random number $r$ as input, outputs signature $s$ for message $m$.
    
    \item $\text{BlindVerify}(PK, m, s) \rightarrow \{0,1\}$: Takes $PK$, $m$, and $s$ as input, outputs 1 if $s$ is valid for $m$, otherwise 0.
\end{enumerate}

A secure blind signature scheme realizes two security properties: unforgeability and blindness~\cite{turan2020tmps}.
Unforgeability ensures that only the signer can generate valid blind signatures.
Blindness ensures that the signer cannot know the message content corresponding to the blind signature she has signed.

}



\vspace{-2mm}
\section{SA$^2$FE Design}
\label{sec:solution}
\noindent
In this section, we first design the puzzle primitive and present its two implementations. 
\revat{Table~\ref{tab1:notation} lists SA$^2$FE notations.}

\vspace{-5mm}
\subsection{Puzzle Design}
\noindent 
A \textbf{rerandomizable puzzle} is constructed for a specific solution.
Anyone with the puzzle and the solution can verify that the solution is correct.
It allows anyone with neither the solution nor any secret used when constructing the puzzle to rerandomize the puzzle without changing the solution.
\revct{
\begin{definition}
A \textbf{rerandomizable puzzle} scheme consists of the following four algorithms:
\begin{enumerate}
    \normalfont{
    \item $\text{PuzzleSetup}(1^\lambda) \!\rightarrow\! params$: 
    Initialize puzzle parameters. 
    \item $\text{PuzzleGen}(m) \!\rightarrow\! puzzle$: Generate a $puzzle$ given a solution message $m$.
    \item $\text{PuzzleMatch}(m,puzzle) \!\rightarrow\! \{0,1\}$: 
    Check if $m$ is the solution to $puzzle$.
    \item $\text{PuzzleRerandomize}(puzzle) \!\rightarrow\! new\_puzzle$:
    Rerandomize $puzzle$ without changing the solution $m$, such that $new\_puzzle$ is unlinkable to $puzzle$.\myendbox

    }
\end{enumerate}
\end{definition}
}

\revct{
The following properties must be fulfilled:
(a) Correctness: $\text{PuzzleMatch}(m, \text{PuzzleGen}(m) \!=\! 1$ for any $m$; 
(b) Soundness: $\Pr[\text{PuzzleMatch}(\hat{m}, \text{PuzzleGen}(m)) \!=\! 1] \approx 0$ for $\hat{m} \!\neq\! m$; 
(c) Indistinguishability: it is computationally hard to distinguish $\text{PuzzleGen}(m)$ from $\text{PuzzleGen}(\hat{m})$ for any $m \!\neq\! \hat{m}$; 
(d) Unlinkability: given a ${puzzle}$, a  ${new\_puzzle}$ can be generated such that $\text{PuzzleMatch}(m, {new\_puzzle}) \!=\! 1$ and ${new\_puzzle}$ is unlinkable to ${puzzle}$ for any $m$.
}

In the following, we propose two puzzle implementations based on bilinear map and universal re-encryption respectively to realize the rerandomizable puzzle primitive.

\noindent \textbf{Puzzle based on bilinear map.}
Let $\mathbb{G}_1$, $\mathbb{G}_2$ and $\mathbb{G}_T$ denote three 
cyclic groups of a prime order $p$, and $g_1$ and $g_2$ be the generators of group $\mathbb{G}_1$ and $\mathbb{G}_2$ respectively.
\grey{A bilinear map $e:\mathbb{G}_1 \!\times\! \mathbb{G}_2 \!\rightarrow\! \mathbb{G}_T$ satisfying bilinearity, computability, and non-degeneracy can be used to construct the bilinear-based rerandomizable puzzle as follows.}

\begin{enumerate}
    \item $\text{PuzzleSetup}(1^\lambda) \!\rightarrow\! params$: 
    Let $parmas = (g_1, g_2)$.
    \item $\text{PuzzleGen}(m)\!\rightarrow\! puzzle$:
    Generate random factor $r \in \mathbb{Z}_p^*$ where $r \mod m = 0$. Then output $puzzle = (z_{[1]},z_{[2]})$, where $z_{[1]} = g_2^{r/m}$ and $z_{[2]} = g_2^{r}$.
    \item $\text{PuzzleMatch}(m,puzzle) \!\rightarrow\! \{0,1\}$:
    Check if $e(g_1^m,z_{[1]}) \!=\! e(g_1,z_{[2]})$.
    \item $\text{PuzzleRerandomize}(puzzle) \!\rightarrow\! new\_puzzle$: Output $new\_puzzle \!=\! ((z_{[1]})^{r'},(z_{[2]})^{r'})$ with random $r' \!\in\! \mathbb{Z}_p^*$. 
\end{enumerate}

\noindent \textbf{Puzzle based on universal re-encryption.}
ElGamal encryption~\cite{elgamal1985public} is an asymmetric key encryption scheme for public-key cryptography.
\grey{With the public key, any ciphertext can be re-encrypted into an unrelated ciphertext.}
\grey{The universal re-encryption scheme~\cite{golle2004universal} hides the public key by appending a second ElGamal ciphertext encrypting the integer $1$.}
\grey{Leveraging ElGamal's algebraic homomorphism, the second ciphertext can re-encrypt the first without exposing the public key.}
\grey{The universal re-encryption-based puzzle is constructed as follows.}
\begin{enumerate}
    \icdcs{\item $\text{PuzzleSetup}(1^\lambda) \!\rightarrow\! (x,y)$: Output $(x,y=g^x)$, where $x \in \mathbb{Z}_q$,
    $g$ is a generator for a group $\mathbb{G}$ with order $q$.}
    \item $\text{PuzzleGen}(m) \!\rightarrow\! puzzle$: 
    Generate random factor $r \!=\! (r_0,r_1) \!\in\! \mathbb{Z}_q^2$.
    Then output $puzzle = [(\alpha_0,\beta_0);(\alpha_1,\beta_1)] = [(my^{r_0},g^{r_0});(y^{r_1},g^{r_1})]$.
    
    \item $\text{PuzzleMatch}(m, puzzle) \!\rightarrow\! \{0,1\}$:
    Verify $\alpha_0, \beta_0, \alpha_1, \beta_1 \in \mathbb{G}$, return $0$ if invalid.
    Compute $m_0 \!=\! \alpha_0/\beta_0^x$ and $m_1 \!=\! \alpha_1/\beta_1^x$. 
    If $m_1 \!=\!1$, return $1$ if $m_0=m$.

\icdcs{
    \item $\text{PuzzleRerandomize}(puzzle) \!\rightarrow\! new\_puzzle$: 
    Output $[(\alpha'_0,\beta'_0);(\alpha'_1,\beta'_1)] \!=\! [(\alpha_0 \alpha_1^{r'_0}, \beta_0 \beta_1^{r'_0});(\alpha_1^{r'_1}, \beta_1^{r'_1})]$ with random factor $r' = (r'_0,r'_1) \in \mathbb{Z}_q^2$.
}

\end{enumerate}

\revat{These two designs are based on different assumptions and have different overheads.}
The bilinear map puzzle and the universal re-encryption puzzle are based on Decisional Bilinear Diffie-Hellman assumption (DBDH)~\cite{green2007identity} and Decisional Diffie-Hellman assumption (DDH)~\cite{tsiounis1998security}, respectively. 
\revat{We evaluate the overhead of these two designs in the \S\ref{sec:eval}.}

\begin{table}[tt]
    \centering
    \caption{\revat{Notation table}}
    \vspace{4pt}
    \begin{tabular}{|c|p{27mm}|c|p{20mm}|}
    \hline
    \textbf{Symbol} & \textbf{Definition} & \textbf{Symbol} & \textbf{Definition} \\
    \hline
    
    $\lambda$ & Security parameter &$k_s$ & Service key \\
    $pk_p, sk_p$ & FA public \& secret key & $s\_type$ & Service type \\
    $pk_s, sk_s$ & SP public \& secret key & $s\_alg$ & Service program \\
    $m_1, m_2$ & Random messages & $Z_{{\sf ID}_u}$ & User puzzle list\\
    $p\_map$ & Puzzle mapping table &  &  \\

    \hline
    \end{tabular}
    \label{tab1:notation}
\end{table}

\revct{
\begin{theorem}
    The bilinear map-based puzzle and the universal re-encryption-based puzzle satisfy the properties of correctness, soundness, indistinguishability, and unlinkability as defined for a randomizable puzzle scheme.
    \end{theorem}
    
    \begin{proof}
    We provide proofs for each property for both the bilinear map-based and universal re-encryption-based puzzles.

    Correctness: The correctness of the bilinear map-based puzzle follows from $e(g_1^m, g_2^{r/m}) = e(g_1, g_2^{r})$. For the universal re-encryption-based puzzle, correctness follows from $\alpha_1/\beta_1^x = y^{r_1}/g^{r_1x}=1$ and $\alpha_0/\beta_0^x = m y^{r_0}/g^{r_0x}=m$.

    Soundness: The soundness of the bilinear map-based puzzle holds because, for any $\hat{m} \!\neq\! m$, $e(g_1^{\hat{m}}, g_2^{r/m}) \!\neq\! e(g_1, g_2^{r})$. 
    For the universal re-encryption-based puzzle, soundness is ensured since $\hat{m} y^{r_0}/g^{r_0x} \!\neq\! m$ for any $\hat{m} \!\neq\! m$.

    Indistinguishability: The bilinear map-based puzzle's indistinguishability relies on the DBDH assumption, ensuring $(g_2^{r/m},g_2^{r})$ and $(g_2^{\hat{r}/\hat{m}},g_2^{\hat{r}})$ with $e(g_1^{m},g_2^{r/m})\!\!=\!\!e(g_1,g_2^{r})$ and $e(g_1^{\hat{m}},g_2^{\hat{r}/\hat{m}})\!\!=\!\!e(g_1,g_2^{\hat{r}})$ are indistinguishable.
    For the universal re-encryption-based puzzle, it derives from the DDH assumption, ensuring $[(m(g^x)^{r_0},g^{r_0});((g^x)^{r_1},g^{r_1})]$ and $[(\hat{m}(g^x)^{\hat{r_0}},g^{\hat{r_0}});((g^x)^{\hat{r_1}},g^{\hat{r_1}})]$ are indistinguishable.

    Unlinkability: The bilinear map-based puzzle's unlinkability also relies on the DBDH assumption, ensuring that $(g_2^{r/m},g_2^{r})$ and $(g_2^{rr'/m},g_2^{rr'})$ with $e(g_1^{m},g_2^{r/m})\!=\!e(g_1,g_2^{{r}})$ and $e(g_1^{m},g_2^{rr'/m})\!=\!e(g_1,g_2^{{r}r'})$ remain indistinguishable, ensuring unlinkability.
    For the universal re-encryption-based puzzle, unlinkability similarly derives from the DDH assumption, ensuring that $[(m(g^x)^{r_0},g^{r_0});((g^x)^{r_1},g^{r_1})]$ and $[(m(g^x)^{r_0}(g^x)^{r_1r'_0},g^{r_0}g^{r_1r'_0});((g^x)^{r_1r'_1},g^{r_1r'_1})]$ are indistinguishable, ensuring unlinkability between the original puzzle and the rerandomized puzzle.
    \end{proof}
    
    \vspace{-2mm}
}
\revat{We will integrate this puzzle primitive into the offloading scheme, enabling the ES to register at the BS and allowing users to randomly select an indistinguishable puzzle from the BS's puzzle list to ensure fairness in offloading.}

\vspace{-4mm}
\subsection{System Initialization}

\noindent 
The system initialization phase initializes all parameters and components required for the system, as shown in Algorithm~\ref{alg:setup}.

\begin{algorithm}[tt] \label{setup}

    \caption{System Initialization}
    \label{alg:setup}
    \vspace{-2mm}
    \footnotesize
    \begin{multicols}{2}
    \LinesNumbered
    
    \nonl \textbf{\makebox[\linewidth]{\underline{At Service Provider}}} \\

     $\text{BlindSetup}(1^\lambda) \!\rightarrow\! (pk_s,sk_s)$\; \label{line:blind_setup_1} 
    $\text{SymKeySetup}(1^\lambda) \rightarrow k_s$\; \label{line:sym_setup}

    \nonl \textbf{\makebox[\linewidth]{\underline{At User}}} \\
    Get blind signature public key $pk_s$ \\\nonl out-of-band from the SP\;   \label{line:user_out_of_band}

    \nonl \textbf{\makebox[\linewidth]{\underline{At Financial Authority}}} \\
    $\text{BlindSetup}(1^\lambda) \rightarrow (pk_p,sk_p)$\; \label{line:blind_setup_2} 
    $\text{PuzzleSetup}(1^\lambda) \rightarrow params$. \label{line:puzzle_setup}
    \end{multicols}
    \vspace{0.6mm}
\end{algorithm}

\noindent \textbf{Service setup {(line~\ref{line:sym_setup})}.} 
The SP sets up the service key $k_s$ that will be used for symmetric encryption of service data. A symmetric encryption scheme has three algorithms, SymKeySetup, SymEnc and SymDec. 
We employ symmetric encryption for the serviced data to reduce the overhead of encrypting and decrypting. 
Alternative encryption schemes can be used if they are compatible with the service data.
\revct{Due to the lightweight nature of our framework, the service key can be efficiently updated periodically based on security requirements or manually rotated upon detecting anomalies.}

\noindent \textbf{Blind signature setup (lines~\ref{line:blind_setup_1}, \ref{line:user_out_of_band}$-$\ref{line:blind_setup_2}).}
Blind signature~\cite{fuchsbauer2020blind} is employed to preserve user anonymity during token usage and maintain confidentiality of the service type from the BS, while still enabling service type verification by the ES and FA.

\grey{Each token includes a service-agnostic part for the BS and a service-specific part for the ES and FA, supported by blind signature setups from the SP (line~\ref{line:blind_setup_1}) and FA (line~\ref{line:blind_setup_2}).}
Regarding the service-specific part, 
the service blind signature secret key $sk_s$ is securely kept confidential by the FA after the SP registers with it, as will be shown in Algorithm~\ref{alg:registration}.
And the service blind signature public key $pk_s$ is only accessible to authorized users and ESs who can access the service $s$ (line~\ref{line:user_out_of_band}).
Regarding the service-agnostic part, 
the 
FA invokes $\text{BlindSetup}$ and 
the public key $pk_{p}$ is made publicly available, allowing the BS to verify the service-agnostic part of tokens.

\noindent \textbf{Rerandomizable puzzle setup (line~\ref{line:puzzle_setup}).}
The FA sets up the parameters of rerandomizable puzzle.
For puzzle based on the bilinear map, both $g_1$ and $g_2$ can be published. 
\revct{For puzzle based on universal re-encryption, only $y$ can be made public, while the corresponding $x$ is obtained by authorized users.}

\vspace{-6mm}
\subsection{Registration}

\noindent 
\grey{Registration phase (steps (1)$-$(3) in Fig.~\ref{model}) is in Algorithm~\ref{alg:registration}.}

\noindent \textbf{SP registration (line~\ref{line:SP_register}).}
An SP registers the service symmetric key $k_s$, the service blind signature public key $pk_s$ and secret key $sk_s$ with the FA.
Then the FA can issue and verify the validity of service tokens for service $s$ by using these keys.

\noindent \textbf{Token registration (lines~\ref{line:user_register_start}$-$\ref{line:user_register_end}).}
To access a service $s$, 
a user first needs to acquire tokens for $s$ from the FA, which needs to include both a service-agnostic part for the BS, and a service-specific part for the ES and FA.
\grey{To request a token for service $s$, the user selects random messages $m_1$, $m_2$, and a random factor $r$, then blinds  $m_1$ and $m_2$ using $\text{BlindMsg}$ to generate  blind messages (lines~\ref{line:user_register_start}$-$\ref{line:user_blind_message_end}).}
Then the user sends blinded messages $m'_1$ and $m'_2$ along with payment information to the FA (line~\ref{line:user_send_blinded_message_to_FA}).
After verifying the payment, the FA signs $m'_1$ and $m'_2$ by invoking $\text{BlindSign}$.
Then the FA sends the blinded signatures and the service key $k_s$ back to the user (lines~\ref{line:FA_sign_token_start}$-$\ref{line:FA_sign_token_end}).
After invoking $\text{UnblindSign}$, the user gets a valid $token$ in the format of $(m_1,sig_1;m_2,sig_2)$ (lines~\ref{line:user_get_valid_token_begin}--\ref{line:user_get_valid_token}).

\begin{algorithm}[tt] \label{registration}
    \vspace{-2mm}
    \footnotesize
    \begin{multicols}{2}
    \caption{Registration}
    \label{alg:registration}
    \LinesNumbered
    \SetNoFillComment
    \tcc{SP registration}
    \nonl \textbf{\makebox[\linewidth]{\underline{At Service Provider}}} \\
    Register with FA by sending $k_s$ and $(pk_s,sk_s)$ to the FA. \label{line:SP_register}

    \tcc{Token registration}
    \nonl \textbf{\makebox[\linewidth]{\underline{At User}}} \\
    Select two random messages $m_1$ and $m_2$, and a random number $r$\; \label{line:user_register_start}

    $\text{BlindMsg}(pk_p, m_1, r) \rightarrow m'_1$\;

    $\text{BlindMsg}(pk_s, m_2, r) \rightarrow m'_2$\;\label{line:user_blind_message_end}

    Send request $(s\_type, m'_1,m'_2,$ $payment)$  to the FA\;\label{line:user_send_blinded_message_to_FA}

    \nonl \textbf{\makebox[\linewidth]{\underline{At Financial Authority}}} \\

    \If{registration request is valid   \label{line:FA_sign_token_start}}
    {
    $\text{BlindSign}(sk_p, m'_1) \rightarrow sig'_1$\;
    $\text{BlindSign}(sk_s, m'_2) \rightarrow sig'_2$\;
    Send  $(sig'_1,$ $sig'_2,k_s)$  to  user\; \label{line:FA_sign_token_end}
    }

    \nonl \textbf{\makebox[\linewidth]{\underline{At User}}} \\

    $\text{UnblindSign}(pk_p, sig'_1, r) \!\rightarrow\! sig_1$\;\label{line:user_get_valid_token_begin}
    $\text{UnblindSign}(pk_s, sig'_2, r) \!\rightarrow\! sig_2$\;
    $token =(m_1,sig_1;m_2,sig_2)$;\label{line:user_get_valid_token} \label{line:user_register_end}

    \tcc{ES reg. to SP}
    \nonl \textbf{\makebox[\linewidth]{\underline{At Edge Server}}} \\
    
    Send $(reg\_info, s\_type)$ to SP\; \label{line:ES_register_start}

    \nonl \textbf{\makebox[\linewidth]{\underline{At Service Provider}}} \\

    \If{ES eligibility is verified}{
        Send $(k_s,pk_s,s\_alg)$ to ES\;\label{line:ES_register_at_SP_end}
    }

    \tcc{ES reg. to BS}
    
    \nonl \textbf{\makebox[\linewidth]{\underline{At Edge Server}}} \\

    \For{registered service $s$  \label{line:ES_start_register_at_BS}}{
 
    $\text{PuzzleGen}(h(k_s)) \!\rightarrow\! puzzle$\;
    
    Send $(puzzle, {\sf ID}_{BS})$ to BS\;
    }

    \nonl \textbf{\makebox[\linewidth]{\underline{At Base Station}}} \\
    
    \While{$puzzle$ from ${\sf ID}_{ES}$}
    {
    Store $puzzle$ in puzzle list $Z$\;
    Insert $(puzzle, {\sf ID}_{ES})$ into a mapping table $p\_map$. \label{line:ES_register_end}
    }

    \end{multicols}
    \vspace{0.4mm}
\end{algorithm}

\noindent \textbf{ES registration to SP (lines~\ref{line:ES_register_start}$-$\ref{line:ES_register_at_SP_end}).}
The ES first generates a registration request and sends it to the SP.
The SP checks service type and verifies if the ES with $reg\_info$ can provide the service $s$.
\grey{If the ES is eligible, the SP sends back service key $k_s$, blind signature key $pk_s$, and service program $s\_alg$.}

\noindent \textbf{ES registration to BS  (lines~\ref{line:ES_start_register_at_BS}$-$\ref{line:ES_register_end}).}
\grey{After receiving the service information, each ES registers with the BS as a candidate to provide service.}
\grey{This process enables ES discovery and equips users with information for fair ES selection.}
For registered service $s$, the ES generates a puzzle by invoking $\text{PuzzleGen}(h(k_s))$, 
where $h(\cdot)$ is a one-way hash function to protect the service key.
Upon receiving the puzzles from ESs, 
the BS stores puzzles in list $Z$ and creates a mapping table $p\_map$ associating each puzzle with ES identity ${\sf ID}_{ES}$.

\revct{To adapt to varying ES capabilities, fairness can be extended by allowing higher-capability ESs to register multiple puzzles, increasing their selection likelihood proportionally to their capacity while maintaining efficiency and security.}

\vspace{-4.0mm}
\subsection{Offloading}

\noindent
Algorithm~\ref{alg:offloading} shows the detailed offloading phase, which corresponds to steps (4)$-$(6) in Fig. \ref{model}.

\begin{algorithm}[tt] \label{alg:offloading}
    \vspace{-2mm}
    \footnotesize
    \begin{multicols}{2}
    \caption{Offloading}
    
    \LinesNumbered
    \SetAlgoHangIndent{0em}

    \nonl \textbf{\makebox[\linewidth]{\underline{At User}}} \\
   
    Generate offloading request $(token,$ ${\sf ID}_{BS})$ and send it to  BS\;\label{line:offloading_start}

    \nonl \textbf{\makebox[\linewidth]{\underline{At Base Station}}} \\

    Form puzzle list $Z_{{\sf ID}_u}$ with the latest versions for user ${\sf ID}_u$\; \label{line:BS_construct_init_puzzle_list}

        \For {$z \in Z_{{\sf ID}_u}$\label{line:BS_rerandomize_puzzle_start}}
    {
        $\text{PuzzleRerandomize}(z) \rightarrow z'$\;
        Replace $z$ with $z'$\;
        Record $(z',{\sf ID}_{ES})$ in $p\_map$\;
    }
    Permute $Z_{{\sf ID}_u}$ to a new puzzle list $Z_{{\sf ID}_u}'$ and send it to user\; \label{line:BS_rerandomize_puzzle_end}

    \nonl \textbf{\makebox[\linewidth]{\underline{At User}}} \\
    Candidate puzzle list $Z_c = \emptyset$\; \label{line:user_puzzle_match_start}

    \For{$z \in Z_{{\sf ID}_u}'$ and $\normalfont \text{PuzzleMatch}(h(k_s),z) = 1$}
    {
            $Z_c = Z_c \cup \{z\}$\;
    }
    Randomly pick a puzzle $z_u\in Z_c$\; \label{line:user_puzzle_match_end}
    Send $(z_u,ct = \text{SymEnc}(k_s, $ $(s\_type, data)))$ to BS\; \label{line:user_send_request_data_to_BS}

    \nonl \textbf{\makebox[\linewidth]{\underline{At Base Station}}} \\

     \If{$\normalfont \text{BlindVerify}(pk_p, m_1, sig_1)=1$  \textbf{and} $token$ is unseen\label{line:BS_check_double_spend}}
    {
        \hspace{-2.5mm}\If{$z_u \in Z_{{\sf ID}_u}'$ and none of puzzles in $Z_{{\sf ID}_u}'$ have been used}
        {Send $(token, ct)$ to ES according to $p\_map$\; \label{line:BS_forward_to_ES}}

    }

    \nonl \textbf{\makebox[\linewidth]{\underline{At Edge Server}}} \\

     \For{$\forall pk_s$ held by the ES\label{line:ES_find_corresponding_service_key_start}}
    {
        \hspace{-2.5mm}\If{$\normalfont \text{BlindVerify}(pk_s, m_2, sig_2)$ $=1$}{
            \label{line:ES_find_corresponding_service_key_end}
            $\text{SymDec}(k_s, ct) \rightarrow data$\; \label{line:ES_settle_service_type}
            {
                Send $resp = \text{SymEnc}(k_s,$ $s\_alg(data))$ to BS\; \label{line:ES_generate_response_encrypt}
                
            }
            \textbf{break}\;
        }
    }

    \nonl \textbf{\makebox[\linewidth]{\underline{At Base Station}}} \\
    Forward the $resp$ to user\; \label{line:BS_forward_response}

    \nonl \textbf{\makebox[\linewidth]{\underline{At User}}} \\
    $\text{SymDec}(k_s, resp) \rightarrow resp\_data$. \label{line:user_get_service_data}
    \end{multicols}
\end{algorithm}

\grey{The user initiates offloading by sending a request to the BS (line~\ref{line:offloading_start}).}
Upon receiving the $init\_request$ from the user with ${\sf ID}_U$, 
the BS performs the following actions (lines~\ref{line:BS_construct_init_puzzle_list}$-$\ref{line:BS_rerandomize_puzzle_end}): 
the BS first constructs a puzzle list $Z_{{\sf ID}_u}$ that contains all the latest version puzzles.
Then the BS re-randomizes the puzzles $z \!\in\! Z_{{\sf ID}_u}$ by invoking $\text{PuzzleRerandomize}(z) \!\rightarrow\! z'$ and replaces $z$ with $z'$. Also, the BS records $(z',{\sf ID}_{ES})$ in $p\_map$.
\grey{The BS sends the permuted puzzle list $Z_{{\sf ID}_u}'$ to the user.}
\revct{The BS re-randomizes and permutes puzzles to ensure fairness, preventing users from targeting specific ESs. 
While acting as a man-in-the-middle, it cannot infer service types, with protocol adherence incentivized by its reliance on reputation.
}

The user, upon receiving $Z_{{\sf ID}_U}'$, proceeds with the following steps (lines~\ref{line:user_puzzle_match_start}$-$\ref{line:user_send_request_data_to_BS}):
for each puzzle received from the BS, the user matches it with the service key for the desired service, constructing a sub-list $Z_c$ of matching puzzles.
The user then randomly selects one puzzle $z_u$ from $Z_c$.
The user encrypts request $data$ with service key $k_s$ to get the ciphertext $ct$, and constructs a message $(z_u, ct)$ which is then sent to the BS.

Upon receiving the user's offloading request, 
the BS performs the following steps  (lines~\ref{line:BS_check_double_spend}$-$\ref{line:BS_forward_to_ES}):
the BS checks the validity of the service-agnostic part of the token $(m_1, sig_1)$ by invoking $\text{BlindVerify}$ and ensuring full token has not been used before.
The BS validates $z_u$ by confirming its presence in the unique puzzle list $Z_{{\sf ID}_U}'$ provided to user ${\sf ID}_U$ and ensuring that none of the puzzles in $Z_{{\sf ID}_U}'$ have been used before.
The BS then finds the ES corresponding to $z_u$ in its mapping table $p\_map$, and forwards $(token,ct)$ to the ES.

Then the ES verifies the service-specific part of the token by invoking $\text{BlindVerify}(pk_s, m_2, sig_2)$.
For ESs offering multiple services, they need to check the service blind signature public key, $pk_s$, associated with each service type to find the corresponding service key, $k_s$ (lines~\ref{line:ES_find_corresponding_service_key_start}$-$\ref{line:ES_find_corresponding_service_key_end}).
If the token check passes, the ES decrypts $ct$ (line~\ref{line:ES_settle_service_type})
and 
continues generating response data for the user using service algorithm $s\_alg$ on $data$.
The ES encrypts the response data with $k_s$ and sends it to the BS (line~\ref{line:ES_generate_response_encrypt}). The BS forwards the encrypted $resp$ to the user (line~\ref{line:BS_forward_response}), allowing the user to decrypt it with $k_s$ (line~\ref{line:user_get_service_data}).
The user and the ES then engage in actual service offloading through the BS until the offloading request is fulfilled.

\vspace{-4mm}
\subsection{Payment Claim}

\begin{algorithm}[tt] \label{alg:audit}
    \caption{Payment Claim}
    \vspace{-2mm}
    \footnotesize
    \begin{multicols}{2}
    \LinesNumbered
    \DontPrintSemicolon

    \While{$(s\_type,token)$ from ES \label{line:ES_token_claim_start}}
    {
    \hspace{-2.5mm}\lIf{$\normalfont \text{BlindVerify}(pk_s, m_2, sig_2)=1$ \textbf{and} token valid for $s\_type$}
            {FA pays to ES and SP;  \label{line:ES_token_claim_end}}
    }

    \While{$(token)$ from BS\label{line:BS_token_claim_start}}
    {
    \hspace{-2.5mm}\lIf{$\normalfont \text{BlindVerify}(pk_p, m_1, sig_1) = 1$ \textbf{and} token not double spent}
        {FA pays to  BS. \label{line:BS_token_claim_end}}
    }

    \end{multicols}
    \vspace{2mm}
\end{algorithm}

\noindent
The payment claim process in Fig.~\ref{model} step (7) is shown in Algorithm~\ref{alg:audit}.
The process is the same for a BS or an ES, except that the public key used to verify the blind signature is different, and for ES the checking needs to additionally verify $s\_type$. 
Upon receiving a token claim request, the FA first checks whether the token has been double spent. 
It then proceeds to verify the token's validity by invoking the function 
$\text{BlindVerify}$.
Upon successful token verification, the FA pays the corresponding tokens to the SP, BS, and ES as per the established contract, ensuring accountability for token claims.

\revct{In practice, beyond presenting a valid $token$ as payment proof, the BS and ES can incorporate other types of proof of service to claim rewards. For instance, they can utilize existing edge service verification schemes that leverage cryptography to generate tamper-proof service proofs~\cite{wang2024veriedge}.}



\vspace{-2mm}
\section{Security Analysis}
\label{sec:security-analysis}
\noindent

\vspace{-5mm}
\subsection{Informal Security Analysis}
\noindent
\noindent
\textbf{Malicious user.}
\grey{Upon receiving the user's offloading request, the BS validates the token's service-agnostic part for FA signature and checks both token parts for prior use.}
Only when the check on $(m_1, sig_1)$ returns valid and the full token has not been seen before, will the BS proceed to construct a puzzle list and share it with the user. 
Full token checking ensures that the user cannot reuse a forged token, for instance, by combining a low-priced BS service-agnostic part with a high-priced service-specific part to double spend the token.
The BS's re-randomization and permutation of the puzzle list
ensure the unlinkability within one round and different rounds of the offloading, so that the user cannot identify a specific ES to disturb the fairness of the offloading process.
Stale puzzle submissions are discarded, thwarting puzzle replay attacks.

\noindent
\textbf{Curious BS on user service request.}
\grey{SA$^2$FE prevents a curious BS from inferring a user's service type by limiting access to sensitive information. The BS can only verify the service-agnostic part of the token, which reveals no service details. User data is encrypted with the service symmetric key, preventing access to the service type or request data. Additionally, puzzles from ESs are indistinguishable and disclose no service-related information.}
\grey{Since the user performs puzzle matching and selection, the BS cannot infer the service type from the list of eligible ESs or from specific ES involvement.}

\noindent
\textbf{Curious FA, SP, BS, and ES  on user identity.}
Due to the blind signature's blindness properties, no one can link the token in a service request to the user's real identity.

\noindent
\textbf{Malicious BS and ES on payment claim.}
If the BS or an ES exaggerates the provided service for extra rewards, the FA will detect invalid or double-spent tokens and reject the claim.

\revat{For a puzzle list of $n$ puzzles, the probability of a malicious user identifying a specific ES is $\frac{1}{n}$. The attack success probability of other attacks within our threat model are negligible unless they violate fundamental cryptographic assumptions.}

\vspace{-4mm}
\subsection{Formal Security Analysis}
\noindent
We next formally analyze SA$^2$FE's security based on the UC framework.
The UC is a widely used simulation-based cryptographic framework for modular security analysis in diverse scenarios, including blockchain~\cite{kate2023flexirand}, federated learning~\cite{hao2023robust}, and quantum key distribution~\cite{ben2005universal}.
It guarantees security even when a secure protocol is composed with an arbitrary set of protocols~\cite{canetti2001universally}.
The definition of UC-security is as follows: 
\begin{definition} \label{def:uc-security}
    \textbf{UC-security~\cite{canetti2001universally}}. Given a security parameter $\lambda$, an ideal functionality $\mathcal{F}$ and a real world protocol $\pi$, we say that $\pi$ securely realizes $\mathcal{F}$ if for any probabilistic polynomial time (PPT) adversary $\mathcal{A}$,  there exists a PPT simulator $\mathcal{S}$ such that for any PPT environment $\mathcal{Z}$, we have 
    $$\mathrm{IDEAL}_{\mathcal{F},\mathcal{S},\mathcal{Z}}\stackrel{\mathrm{c}}{\equiv} \mathrm{REAL}_{\mathcal{\pi},\mathcal{A},\mathcal{Z}}.$$
     The $\stackrel{\mathrm{c}}{\equiv}$ denotes computational indistinguishable.\myendbox
\end{definition}

We denote the ideal functionality of SA$^2$FE as $\mathcal{F}_{\sf{SA^2FE}} = \langle \mathcal{F}_{\sf{register}}, \mathcal{F}_{\sf{offload}}, \mathcal{F}_{\sf{claim}}, \mathcal{F}_{\sf{sig}}, 
\mathcal{F}_{\sf{smt}} \rangle$.
    $\mathcal{F}_{\sf{register}}$ is the ideal functionality of registration phase. 
    $\mathcal{F}_{\sf{offload}}$ models the offloading phase.
    $\mathcal{F}_{\sf{claim}}$ models the token claim process.
\grey{Two helper ideal functionalities $\mathcal{F}_{\sf{sig}}$ and $\mathcal{F}_{\sf{smt}}$~\cite{canetti2001universally} are used to model the digital signature and the secure message transmission channel.}
Following UC framework, we assume that each party interacting with the ideal functionalities has a unique identifier,
and consider a static corruption model where the adversary can corrupt parties at the beginning of the protocol.

\begin{figure}[tt]
\fbox{
\begin{minipage}{0.46\textwidth}
\footnotesize
\textbf{Functionality $\mathcal{F}_{\sf{register}}$}

\centering \underline{\textit{Service provider registration}}
\begin{enumerate}
    \item Upon receiving $(\text{register}, spid, sname, sdata)$ from the SP, $\mathcal{F}_{\sf{register}}$ adds $t_{s} = (spid, sname, sdata, \perp, \perp)$ to $T_{\sf s}$.
    If the $t_{s}$ is already in $T_{\sf s}$, then $\mathcal{F}_{\sf{register}}$ returns $t_{s}$ to the SP.
\end{enumerate}

\centering \underline{\textit{Edge server registration}}

\begin{enumerate}
    \item Upon receiving $(\text{register}, spid, sname, esid)$ from ES, 
    $\mathcal{F}_{\sf{register}}$ checks if $T_{\sf s}$ has an entry $t_{s}\!=\! (spid,sname,\cdot,esid,\cdot)$.
    If yes,
    $\mathcal{F}_{\sf{register}}$ returns $t_{s}$ to $esid$,
    and forwards $(\text{exist},t_{s})$ to $\mathcal{S}$. 
    Otherwise,
    $\mathcal{F}_{\sf{register}}$ sends a message $(\text{register},spid,esid)$ to the SP with $spid$.
    If the SP responds with ``allow'', then $\mathcal{F}_{\sf{register}}$ creates an entry $(spid, sname,\cdot,esid,\cdot)$ in $T_{\sf s}$ and forwards $(\text{successReg},(spid,sname,esid))$ to $esid$ and $\mathcal{S}$.
    Otherwise, $\mathcal{F}_{\sf{register}}$ returns ``fail'' to $esid$ and forwards $(\text{failReg},(spid,sname,esid))$ to $\mathcal{S}$.

    \item Upon receiving $(\text{register},t_{p}=(\cdot,spid,sname,0,\textsf{unused},esid,$ $bsid),bsid)$ from ES, 
    $\mathcal{F}_{\sf{register}}$ forwards  $(\text{register},esid,$ $bsid)$ to  BS with $bsid$.
    If BS responds with ``allow'', then $\mathcal{F}_{\sf{register}}$ updates the entry $(spid, sname,\cdot,esid,bsid)$ in $T_{\sf s}$ and forwards $(\text{successReg},(spid,sname,bsid,esid))$ to $esid$ and $\mathcal{S}$.
    Otherwise, $\mathcal{F}_{\sf{register}}$ returns ``fail'' to $esid$ and forwards $(\text{failReg},(spid,sname,bsid,esid))$ to $\mathcal{S}$.

    \item 
    $\mathcal{F}_{\sf{register}}$ adds $t_{p}=(puzzle_{\sf ideal},spid,sname,0,\textsf{unused},esid,bsid)$ to $T_{\sf p}$, 
    and 
    forwards $(\text{newReg},t_{p})$ to $\mathcal{S}$.
\end{enumerate}

\centering \underline{\textit{User registration}}
\begin{enumerate}
    \item Upon receiving $(\text{register}, spid, sname, payment)$ from user, $\mathcal{F}_{\sf{register}}$ sends a message $(\text{register},spid,sname,payment)$ to FA.
    \item If FA returns ``allow'',  $\mathcal{F}_{\sf{register}}$ adds $t_{\sf t} = (token_{\sf ideal},spid,sname,$ $(\perp,\textsf{fresh}),(\perp,\textsf{fresh}))$ in $T_{\sf t}$.
    Then $\mathcal{F}_{\sf{register}}$ sends the $t_{\sf t}$ to the user and forwards $(\text{newReg},t_{\sf t})$ to $\mathcal{S}$.
    \item 
    Otherwise,  
    $\mathcal{F}_{\sf{register}}$ sends ``fail" to user
    and sends $(\text{failReg},t_{\sf t})$ to $\mathcal{S}$.
\end{enumerate}

\end{minipage}

}
\vspace{0.2em}
\caption{Ideal functionality for registration.}
\label{register_ideal}
\end{figure}

The ideal functionality 
$\mathcal{F}_{\sf{SA^2FE}}$ maintains the internal states in three tables, $T_{\sf s}$, $T_{\sf p}$ and $T_{\sf t}$, to ensure consistency of the real world and the ideal world.
$T_{\sf s}$ consists of entries in the format of $(spid, sname, sdata,esid, bsid)$ about the service. 
The $spid$ denotes the SP identity, $sname$ represents the service name, $sdata$ contains the content or data associated with the service, 
and $esid$ and $bsid$ uniquely identify an ES and a BS respectively.
$T_{\sf p}$ contains puzzle information in the format of $(puzzle_{\sf ideal},spid,sname,ver,f_p,esid,$ $bsid)$. The $puzzle_{\sf ideal}$ is a puzzle in the ideal world. 
$ver$
is a number that indicates the puzzle's version that identifies the set of puzzles corresponding to a user's request. 
The puzzles generated in the same batch have the same version number. 
$f_p \!\in\! \{\textsf{unused},\textsf{used} \}$ is a flag indicating whether the puzzle has been received by the BS from a user,
$esid$ is the ES identity that the puzzle corresponds to, and $bsid$ is the BS identity that the $esid$ is registered with.
$T_{\sf t}$ contains the token information in the format of $(token_{\sf ideal},spid,sname,(esid,f_{es}),(bsid,f_{bs}))$. 
The $token_{\sf ideal}$ is a string that indicates the token in the ideal world. 
$esid$ is the ES identity that received the token and  $f_{es}$ is the flag that indicates the status of the token. 
$f_{es}$ has three options: $\textsf{fresh}$ for a newly initialized, unused token, $\textsf{unclaimed}$ for token received but not yet claimed, and $\textsf{claimed}$ for token already claimed.
$(bsid$, $f_{bs})$ is similarly defined for a BS.


\noindent \textbf{Registration.} \grey{The ideal functionality $\mathcal{F}_{\textsf{register}}$, shown in Fig.~\ref{register_ideal}, handles registration for the SP, BS, and ES by creating entries for services, puzzles, and tokens, ensuring validity and freshness while coordinating with other parties.}

\begin{figure}[t]
\fbox{
\begin{minipage}{0.46\textwidth}
\footnotesize
\textbf{Functionality $\mathcal{F}_{\sf{offload}}$}
\begin{enumerate}
    \item Upon receiving a request $(\text{offRequest},token_{\sf ideal},bsid)$ from the user,  $\mathcal{F}_{\sf{offload}}$ constructs a $p\_list$ which includes all the puzzle filed of the puzzles with $(\cdot,\cdot,\cdot,ver_{\sf newest},\cdot,\cdot,bsid)$ in the $T_{\sf p}$.

    \item 
    $\mathcal{F}_{\sf{offload}}$ re-randomizes each puzzle in $p\_list$ by generating a new puzzle string $puzzle_{\sf new}$, 
    and creates a new entry $t_{p} = (puzzle_{\sf new},spid,sname,ver_{\sf newest}+1,\textsf{unused},esid,bsid)$ in $T_{\sf p}$. $\mathcal{F}_{\sf{offload}}$ sends $(\text{randomizedPuzzle},puzzle_{\sf ideal},puzzle_{\sf new})$ to $\mathcal{S}$.
    \item  $\mathcal{F}_{\sf{offload}}$ constructs a $p\_list'$ for user by collecting all puzzles with $ver_{\sf newest}+1$ in a random order. 
    \red{If the user is malicious,}
    $\mathcal{F}_{\sf{offload}}$ sends $p\_list'$ to $bsid$
    and 
    forwards $(\text{offStart},token,p\_list')$ to $\mathcal{S}$.
    \red{If the user is honest, $\mathcal{F}_{\sf{offload}}$ associates each puzzle in the $p\_list'$ with its corresponding $(spid, sname)$ and sends the list to the user.}
    $\mathcal{F}_{\sf{offload}}$ stores the mapping between $uid$ and the corresponding $ver^*$.

    \item Upon receiving response from the user with $(puzzle_{\sf ideal}',M_{\sf data})$, 
    $\mathcal{F}_{\sf{offload}}$ 
    sends $(\text{newRequest},$ $puzzle_{\sf ideal}',M_{\sf data})$ to $\mathcal{S}$.
    $\mathcal{F}_{\sf{offload}}$ sends $(\text{userAbort},uid)$ to $\mathcal{S}$ when there is no  response.

    \item 
    $\mathcal{F}_{\sf{offload}}$ checks if there is an entry $t_{\sf t} = (token_{\sf ideal}, spid,sname,$ $ (\perp,\textsf{fresh}),(\perp,\textsf{fresh}))$ in $T_{\sf t}$.
    If no such $t_{\sf t}$ exist,  $\mathcal{F}_{\sf{offload}}$ returns ``fail" to user 
    and forwards $(\text{invalidToken},token_{\sf ideal}, bsid)$ to $\mathcal{S}$.
   
    \item If there is such an entry $t_{\sf t}$, then $\mathcal{F}_{\sf{offload}}$ sends $(\text{newRequest},$ $puzzle_{\sf ideal}',M_{\sf data})$ to $bsid$ and $\mathcal{S}$.
    $\mathcal{F}_{\sf{offload}}$ verifies the validity of $puzzle_{\sf ideal}'$ within $p\_list'$. 
    If not, $\mathcal{F}_{\sf{offload}}$ returns ``fail" to user
    and forwards $(\text{invalidPuzzle},puzzle_{\sf ideal}',M_{\sf data})$ to $\mathcal{S}$.
    \item If $puzzle_{\sf ideal}'$ is valid and $f_p$ is $\textsf{unused}$, $\mathcal{F}_{\sf{offload}}$ retrieves the $esid$ corresponding to $puzzle_{\sf ideal}'$ and updates the entries with $ver^*$ to $(\cdot,\cdot,\cdot,ver^*,\textsf{used},esid,bsid)$. Then $\mathcal{F}_{\sf{offload}}$ updates the $T_{\sf t}$ with $(token_{\sf ideal}, spid,sname, (esid,\textsf{fresh}),(bsid,\textsf{unclaimed}))$.
    Otherwise, $\mathcal{F}_{\sf{offload}}$ returns ``fail" to user. And $\mathcal{F}_{\sf{offload}}$ forwards $(\text{invalidPuzzle},puzzle_{\sf ideal}',M_{\sf data})$ to $\mathcal{S}$.
    \item  $\mathcal{F}_{\sf{offload}}$ sends $(token_{\sf ideal}, M_{\sf data})$ to $esid$ and sends $(\text{offToES},$ $token_{\sf ideal},esid,M_{\sf data})$ to $\mathcal{S}$.
    $\mathcal{F}_{\sf{offload}}$ updates the $T_{\sf t}$ with $(token_{\sf ideal},spid,sname, (esid,\textsf{unclaimed}),(bsid,\textsf{unclaimed}))$.
    \item On receiving $(token_{\sf ideal}, M_{\sf data})$ from $\mathcal{F}_{\sf{offload}}$, $esid$ retrieves $sdata$ from $\mathcal{F}_{\sf{offload}}$ 
    and sends $M_{\sf resp} \!\leftarrow\! sdata(M_{\sf data})$ to $\mathcal{F}_{\sf{offload}}$.

    \item $\mathcal{F}_{\sf{offload}}$ forwards the response $M_{\sf resp}$ to the user, $bsid$ and $\mathcal{S}$.

\end{enumerate}

\end{minipage}

}
\vspace{0.2em}
\caption{Ideal functionality for offloading.}
\label{offloading_ideal}
\end{figure}

\begin{figure}[t]
\fbox{
\begin{minipage}{0.46\textwidth}
\footnotesize
    \textbf{Functionality $\mathcal{F}_{\sf{claim}}$}
    \begin{enumerate}
        \item Upon receiving $(\text{claimRequest},bsid,token_{\sf ideal})$ from $bsid$, $\mathcal{F}_{\sf{claim}}$ checks if there is an entry $t_{\sf t} = (token_{\sf ideal}, \cdot,\cdot,(\cdot,\cdot),(bsid,$ $\textsf{unclaimed}))$ in $T_{\sf t}$. 
        If no such $t_{\sf t}$ exist, then $\mathcal{F}_{\sf{claim}}$ returns ``fail" to $bsid$ and forwards $(\text{invalidToken},token_{\sf ideal},bsid)$ to $\mathcal{S}$.
        Otherwise, $\mathcal{F}_{\sf{claim}}$ sets the $f_{bs}$ of $t_{\sf t}$ to $\textsf{claimed}$ and returns ``success" to $bsid$.
        Also, $\mathcal{F}_{\sf{claim}}$ forwards $(\text{successClaimed},token_{\sf ideal},bsid)$ to $\mathcal{S}$.

        \item Upon receiving $(\text{claimRequest},esid,spid,sname, token_{\sf ideal})$ from $esid$, $\mathcal{F}_{\sf{claim}}$ checks if there is an entry $t_{\sf t} \!=\! (token_{\sf ideal}, spid,$ $sname,(esid,\textsf{unclaimed}),(\cdot,\cdot))$ in $T_{\sf t}$. 
        If not,
        $\mathcal{F}_{\sf{claim}}$ returns ``fail" to $esid$ and forwards $(\text{invalidToken},$ $token_{\sf ideal},esid)$ to $\mathcal{S}$.
        Otherwise, $\mathcal{F}_{\sf{claim}}$ sets the $f_{es}$ of $t_{\sf t}$ to $\textsf{claimed}$ and returns ``success" to $esid$.
        Also, $\mathcal{F}_{\sf{claim}}$ forwards $(\text{successClaimed},token_{\sf ideal},esid)$ to $\mathcal{S}$.

    \end{enumerate}
\end{minipage}

}
\vspace{0.2em}
\caption{Ideal functionality for payment claim.}
\label{audit_ideal}
\end{figure}

\noindent \textbf{Offloading.}
\noindent The ideal functionality shown in Fig.~\ref{offloading_ideal} models the offloading process between the user, BS and ES.
\grey{The functionality validates tokens and checks if selected puzzles are valid entries in $T_{\sf p}$. Invalid tokens or puzzles result in a failure message to the user and a notification to $\mathcal{S}$. For valid requests, it updates the status of tokens and puzzles, marks puzzles as used, changes the status of tokens, and manages mappings in $T_{\sf t}$ and $T_{\sf p}$. Finally, $\mathcal{F}_{\textsf{offload}}$ forwards the service response $M_{\sf resp}$ to the user, BS, and $\mathcal{S}$.}

\noindent \textbf{Payment Claim.}
\grey{Fig.~\ref{audit_ideal} shows the payment claim ideal functionality, where $\mathcal{F}_{\sf{claim}}$ verifies token validity, prevents double spending by checking $f_{bs}$ or $f_{es}$ status, rewards $bsid$ or $esid$, and updates $f_{bs}$ or $f_{es}$ to mark the token as claimed.}


\begin{theorem} \label{thm:UC}
    Let $\mathcal{A}$ and $\mathcal{S}$ be a PPT adversary and a simulator in the real world and the ideal world, respectively.
    SA$^2$FE securely realizes $\mathcal{F}_{\textsf{SA$^2$FE}}$ for any PPT environment $\mathcal{Z}$.
\end{theorem}


\begin{proof}

We design a series of games, where each game differs slightly from the previous one but remains indistinguishable from the view of the PPT environment $\mathcal{Z}$.

\noindent \textbf{Game 0:} This is the real world protocol SA$^2$FE that interacts directly with the environment $\mathcal{Z}$ and adversary $\mathcal{A}$.

\noindent \textbf{Game 1:} This game is identical to Game 0 except that the real-world communication channel is replaced by  $\mathcal{F}_{\textsf{smt}}$.

\begin{lemma} \label{game0_game1}
    For any $\mathcal{A}$ and $\mathcal{Z}$, there exists an $\mathcal{S}$ such that the view of $\mathcal{Z}$ in Game 1 is indistinguishable from its view in Game 0, i.e.,
    $
    \normalfont \textsf{Exec}_{\text{Game0},\mathcal{Z}} \approx \textsf{Exec}_{\text{Game1},\mathcal{Z}}.
    $
\end{lemma}

\begin{proof}
\grey{$\mathcal{S}$ can run $\mathcal{S}_{\textsf{smt}}$ for $\mathcal{F}_{\textsf{smt}}$ to achieve the indistinguishability between the real world and ideal world.}
\vspace{-2mm}
\end{proof}

\noindent \textbf{Game 2:} 
Let $\mathcal{S}$ have access to the output of both honest parties and the adversary $\mathcal{A}$. 
Then $\mathcal{S}$ tries to simulate the protocol with the help of $\mathcal{F}_{\textsf{SA$^2$FE}}$.
In the real world, $\mathcal{A}$ can corrupt the entities. Subsequently, all incoming and outgoing messages of the corrupted party go through $\mathcal{A}$.
In the ideal world, $\mathcal{S}$ has the ability to corrupt entities and inform $\mathcal{F}_{\textsf{SA$^2$FE}}$ accordingly. In the subsequent process, $\mathcal{F}_{\textsf{SA$^2$FE}}$ will discard all messages from the corrupted party and treat $\mathcal{S}$ as the corrupted party.

\begin{lemma} \label{game1_game2}
    $
    \normalfont \textsf{Exec}_{\text{Game1},\mathcal{Z}} \approx \textsf{Exec}_{\text{Game2},\mathcal{Z}}
    $.
\end{lemma}

\begin{proof}

    Simulator $\mathcal{S}$ obtains setup information from SP and FA. Upon receiving registration requests from SP, ES, and user, $\mathcal{S}$ has sufficient information to generate messages acceptable to $\mathcal{F}_{\sf{register}}$. This allows $\mathcal{F}_{\sf{register}}$ to update the internal tables $T_{\sf s}$, $T_{\sf p}$, $T_{\sf t}$ accordingly.
    Specifically, at ES registers to BS stage, $\mathcal{S}$ records the map between the $puzzle$ from real world and the $puzzle_{\sf ideal}$ in the ideal world.
    At token registration stage, $\mathcal{S}$ maintains a local map $\mathcal{R}$ that associates $token$ with $token_{\sf ideal}$.

    For offloading and payment claim, considering the threat model in \S\ref{threat_model}, the following cases need to be tackled.

\subsubsection{Wrong token from user side}
 Consider a user who tries to double spend a token or use a forged token to get the service from multiple ESs.
    \grey{If the user generates a fake token that deceives the BS, the unforgeability of the blind signature is violated.}
    In the real world, 
    the BS validates user's $token$ and rejects the request for invalid or double-spent tokens.
    \grey{In the ideal world, $\mathcal{S}$ checks and updates $\mathcal{R}$ to reject tokens where $f_{bs}$ is not $\textsf{fresh}$.}
    \grey{$\mathcal{S}$ creates a nonexistent $token_{\sf ideal}$ in $T_{\sf t}$ for an invalid token, causing $\mathcal{F}_{\sf{offload}}$ to return ``invalidToken'' back.}

    \subsubsection{User exploits a specific ES}
    A user linking two re-randomized puzzles will lead to a violation of the DBDH/DDH assumptions.
    \grey{The user can only exploit a specific ES by reusing the same $z_u$ as before.}
    If BS detects that $z_u$ has been used before, it rejects the offloading request.
    In the ideal world, $\mathcal{S}$ retrieves $puzzle_{\sf ideal}'$ of $z_u$ and sends it to $\mathcal{F}_{\sf{offload}}$. 
    If $\mathcal{F}_{\sf{offload}}$ finds that $puzzle_{\sf ideal}'$ for $uid$ is not with $ver^*$ or $puzzle_{\sf ideal}'$ in $T_{\sf p}$ has $f_p \!=\! \textsf{used}$, $\mathcal{F}_{\sf{offload}}$ returns an ``invalidPuzzle'' message to $\mathcal{S}$. $\mathcal{S}$ then responds to the user with ``fail''.

    \subsubsection{BS/ES exaggerates service for rewards}
    \grey{If the BS/ES uses an invalid token or double spends it to claim extra rewards, the FA will detect it and reject the request.}
    In the ideal world, if $\mathcal{S}$ generates a $token_{\sf ideal}$ not existing in $T_{\sf t}$ for an invalid $token$ or the $f_{bs}$/$f_{es}$ of $token_{\sf ideal}$  is ${\sf claimed}$, $\mathcal{F}_{\sf{claim}}$ will return an ``invalidToken'' message to $\mathcal{S}$.

\subsubsection{Curious BS on user service request}
\grey{For puzzle registration from an ES, upon receiving $(puzzle, {\sf ID}_{BS})$, $\mathcal{S}$ forwards it to the BS.}
\grey{The rerandomized puzzle generated by the BS requires no translation by $\mathcal{S}$.}
\grey{For the offloading request, selected puzzle and $(z_u, ct)$ from the user, $\mathcal{S}$ directly forwards the user's output to the BS.}
\grey{As $\mathcal{S}$ only performs forwarding, the real and ideal worlds are indistinguishable.}

\subsubsection{Curious FA, SP, BS, ES on user identity}
The blind signature proof process resembles the ticket request process in~\cite{turan2020tmps}, which has been proved to be UC-secure.
\vspace{-2mm}
\end{proof}

\noindent \textbf{Game 3:} 
In this game, the simulator $\mathcal{S}$ cannot directly communicate with honest parties. Instead, $\mathcal{S}$ needs to generate the outputs of the honest parties to the adversary $\mathcal{A}$.

\begin{lemma} \label{game2_game3}
    $
    \normalfont \textsf{Exec}_{\text{Game2},\mathcal{Z}} \approx \textsf{Exec}_{\text{Game3},\mathcal{Z}}
    $.
\end{lemma}

\begin{proof}
    $\mathcal{S}$ generates the system parameters and keys on behalf of the honest parties.
    In the ideal world, upon receiving a registration request from the corrupted/honest ES or user, $\mathcal{S}$ constructs a corresponding message and sends it to $\mathcal{F}_{\textsf{register}}$. 
    $\mathcal{S}$ generates real-world $puzzle$, $token$, and other information using these parameters and keys, 
    and then sends them to $\mathcal{A}$.

\begin{table*}[ht]
    \centering
    \caption{Communication Cost and Execution Time of SA$^2$FE}
    \vspace{0.5em}
    \label{tab:merged_table}
    \begin{threeparttable}
    \begin{tabular}{c | l r | l r }
        \hline
        & \multicolumn{2}{c|}{\textbf{Communication Cost\tnote{*}}} & \multicolumn{2}{c}{\textbf{Execution Time}} \\
            \cline{2-5}
            \textbf{Description} & \textbf{Message} & \textbf{Size (bytes)} & \textbf{Step} & \textbf{Time (ms)} \\
            \hline
       \multirow{2}{*}{Setup} & \multirow{2}{*}{-} & \multirow{2}{*}{-} & SP setup & 371.30 \\
            &  &  & FA setup & 392.42 \\
        \hline
        SP registration & SP to FA registration request & 2012 & SP registration & 59.74 \\
        \hline
        \multirow{3}{*}{ES registration} & \multirow{1}{*}{ES to SP registration response} & \multirow{1}{*}{748} & ES registered to SP & 2.62 \\
            
            & \multirow{2}{*}{ES to BS registration request} & \multirow{2}{*}{275} & ES puzzle generation time & 3.12 \\
            & & & {ES registered to BS} & {1.57} \\
            
         \hline 

        \multirow{5}{*}{User registration}& \multirow{3}{*}{User to FA registration request} & \multirow{3}{*}{1317} & User blinded token & 1.71 \\
         
         & \multirow{4}{*}{User to FA registration response} & \multirow{4}{*}{1626} & FA signed token & 18.65 \\
         
         &  &  & User got token & 58.89 \\
         &  &  & User unblinded token & 0.84 \\
         &  &  & User verified token & 1.73 \\
        \hline
        \multirow{4}{*}{Offloading}
        & User initial offloading request & 1310 & BS rerandomized puzzles & 55.92 \\
         & User received puzzle list & 8163 & User got response of initial offloading request & 62.19 \\
         & \multirow{2}{*}{User generated service request} & \multirow{2}{*}{309} & User selected puzzle & 237.67  \\
         & & & User got service response & 6.88 \\
        \hline

        \multirow{2}{*}{Payment claim}
        & BS token claim request & 1310 & BS claimed token & 2.24 \\
         & ES token claim request & 1312 & ES claimed token & 2.67 \\
        \hline
    \end{tabular}
    \begin{tablenotes}
        \item[*] Sizes of trivial text messages such as ``success'' or ``fail'' are omitted. The message size excludes the service content data ciphertext.
    \end{tablenotes}
    \end{threeparttable}
\end{table*}

\begin{table}[tt]
    \centering
    \caption{Evaluation Platforms}
    \vspace{0.5em}
    \label{tab:platform}
    \begin{tabular}{|>{\centering\arraybackslash}m{1.5cm}|>{\centering\arraybackslash}m{2.6cm}|>{\centering\arraybackslash}m{1.9cm}|>{\centering\arraybackslash}m{1.1cm}|}
        \hline
        \textbf{Platform} & \textbf{CPU} & \textbf{OS} & \textbf{Memory} \\
        \hline
        HWI-AL00 Phone & Hisilicon Kirin 960 2.36GHz, 8 cores & Android 8.0.0 (ARM) & 6GB \\
        \hline
        Raspberry Pi 4 Model B & Broadcom BCM2835 700MHz, 4 cores & Ubuntu 22.10 (ARM) & 3.7GB \\
        \hline
        Laptop & Intel Core i7-7700HQ 2.80GHz, 8 cores & 64-bit Windows (x86) & 24GB \\
        \hline
        Desktop & AMD Ryzen 3945WX 4.0GHz, 12 cores & 64-bit Ubuntu (x86) & 256GB \\
        \hline
    \end{tabular}
\end{table}

\grey{For offloading requests from corrupted users, $\mathcal{S}$ rerandomizes the $puzzle$ and records its mapping to $puzzle_{\sf new}$.}
\grey{This enables $\mathcal{S}$ to construct real-world puzzle lists indistinguishable within and across batches by using universal re-encryption or bilinear mapping to generate real-world puzzles matching ideal-world random strings, which it then sends to $\mathcal{A}$.}

   For the payment claim protocol, when a corrupted BS/ES claims the reward, $\mathcal{S}$ validates the $token$ and generates corresponding $token_{\sf ideal}$, and sends it to $\mathcal{F}_{\sf{claim}}$. After receiving feedback, $\mathcal{S}$ generates corresponding messages for $\mathcal{A}$.
    \vspace{-2mm}
\end{proof}

Combining Lemmas~\ref{game0_game1}--\ref{game2_game3} proves Theorem~\ref{thm:UC}.
\end{proof}




\vspace{-4mm}
\section{Performance Evaluation}
\label{sec:eval}

\subsection{Implementation and Experiment Settings}
\noindent
We used the gRPC framework (v.1.51.3)~\cite{grpc} to implement the communication between different parties. All protocols were implemented in Python.
We used RSA blind signature for blindly signing the tokens, and used AES in CBC mode for symmetric key encryption, both implemented in the Crypto library (v.3.17)~\cite{crypto}.
For the puzzle based on bilinear map,
we used the pairing library from Charm Crypto (v.0.50)~\cite{charm_crypto} and used the SS512 curve on x86 platform, and the PBC library (v0.5.14) (in C language) on ARM CPU platform.
The puzzle based on universal re-encryption was implemented based on the ElGamal encryption scheme~\cite{elgamal_implement}.
SHA-256 was used as the hash function in the protocol.

\icdcs{We evaluated the performance of SA$^2$FE on four platforms as shown in Table~\ref{tab:platform}}.
By default, the user client was run on the Raspberry Pi, while the other parties were run on the desktop.

\begin{figure}[tt]
    \centering
    \includegraphics[width=0.31\textwidth]{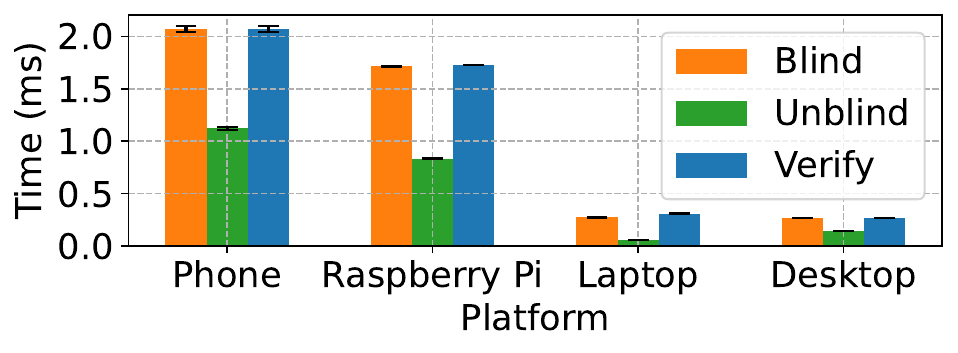}
    \caption{Token registration (blind signature) delay.}
    \label{fig:blind_signature}
\end{figure}

\vspace{-3mm}
\subsection{Evaluation Results}
\label{Ev_2}

\noindent We evaluated the performance of SA$^2$FE by analyzing both the communication overhead and execution time.
\revct{Table~\ref{tab:merged_table} presents the evaluation results, showing the added overhead of our solution compared to ordinary offloading, where service requests are directly offloaded from the user to a edge server without any security guarantees.}
We conducted a statistical analysis on 1,000 runs of SA$^2$FE on different platforms to calculate the average time taken for each step. 
The default number of puzzles was set to 30.

As shown in the communication cost part of Table~\ref{tab:merged_table}, all message sizes were below 9KB. 
The practical execution time part of Table~\ref{tab:merged_table} provides an overview of the delay associated with each step.
SP setup and FA setup phases had the longest delays, which should only be executed once when the system initializes.
The puzzle generation overhead for ES was relatively small, with merely a 3.12ms overhead.
At the user end, the most significant delay during the offloading phase occurred when selecting a puzzle. It took approximately 237.67ms to match and randomly choose one from
30 available puzzles. 
This selection process occurs only once at offloading initiation,
while the duration of a single service session can last for a considerable time, such as minutes to hours~\cite{farhadi2021service,aral2018dependency}.

\begin{figure}[tt]
    \centering
    \includegraphics[width=0.31\textwidth]{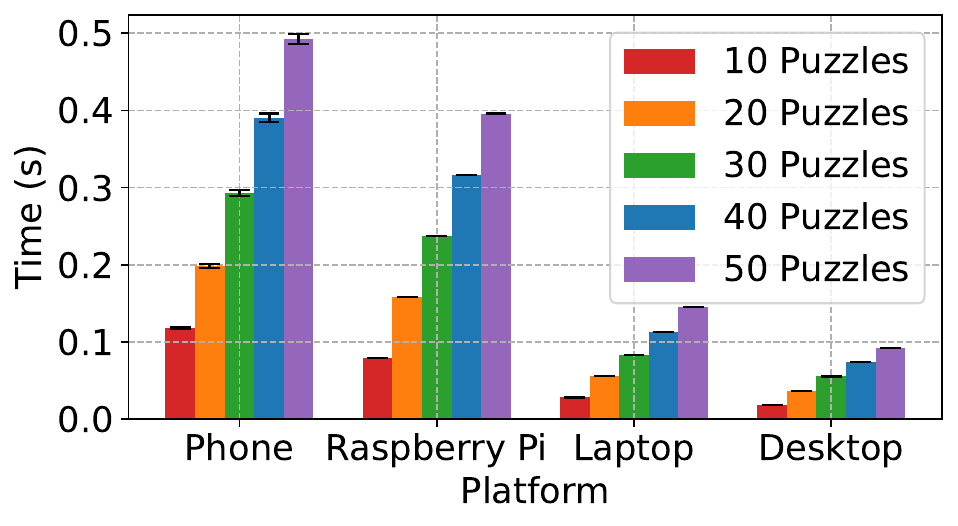}
    \caption{Bilinear map puzzle matching delay.}
    \label{fig:bilinear}
    \vspace{1mm}
\end{figure}

\begin{figure}[tt]
    \centering
    \includegraphics[width=0.31\textwidth]{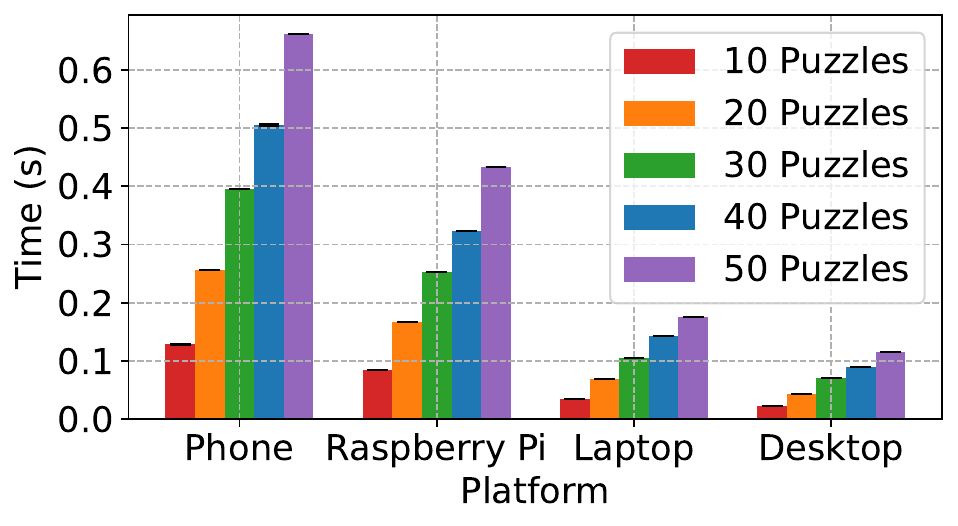}
    \caption{Universal re-encryption puzzle matching delay.}
    \label{fig:URE}
    \vspace{1mm}
\end{figure}

To evaluate the overhead on the user side of SA$^2$FE, we collected the computation delays on user registration and puzzle matching processes on four different platforms. The results are presented in Figs.~\ref{fig:blind_signature}$-$\ref{fig:URE}, with error bars representing the 95\% confidence interval obtained from running each experiment 1000 times.
\grey{Fig.~\ref{fig:blind_signature} shows the delay experienced by the user during the registration.}
It can be observed that the user-side computation overhead during the user registration process was small. Even on the lowest-performing platform, the average delay for each step was at most 2.07ms.
Figs.~\ref{fig:bilinear} and~\ref{fig:URE} focus on two different puzzle implementation approaches: bilinear map and universal re-encryption. 
These figures illustrate the match delay for various numbers of puzzles on different platforms. 
It can be observed that with better computation capability, the match delay decreased. Additionally, as the number of puzzles increased, the delay in selecting all the suitable puzzles from the puzzle list and randomly choosing one puzzle as the final puzzle also increased. 
\revat{Overall, while DBDH is a stronger assumption compared to DDH, the bilinear map puzzle exhibited slightly lower computation overhead compared to the universal re-encryption puzzle with the current implementation.}



\vspace{-2mm}
\section{Conclusion}
\label{sec:conclusions}
\noindent
In this paper, we proposed SA$^2$FE, an anonymous, auditable and fair service offloading framework for
democratized edge computing ecosystem.
A novel rerandomizable puzzle primitive was introduced to enhance the design of the service offloading by preserving service type privacy and enabling fair and randomized edge server selection.
Additionally, a token-based scheme was proposed to enable access control, maintain user anonymity, protect service type confidentiality, and enable accountable token verification and claiming.
We proved the security of SA$^2$FE based on the UC framework.
The experimental results demonstrated the efficiency of SA$^2$FE in terms of communication and computation overhead.


\vspace{-2mm}
\bibliographystyle{IEEEtran}


\begin{thebibliography}{10}
  \providecommand{\url}[1]{#1}
  \csname url@samestyle\endcsname
  \providecommand{\newblock}{\relax}
  \providecommand{\bibinfo}[2]{#2}
  \providecommand{\BIBentrySTDinterwordspacing}{\spaceskip=0pt\relax}
  \providecommand{\BIBentryALTinterwordstretchfactor}{4}
  \providecommand{\BIBentryALTinterwordspacing}{\spaceskip=\fontdimen2\font plus
  \BIBentryALTinterwordstretchfactor\fontdimen3\font minus \fontdimen4\font\relax}
  \providecommand{\BIBforeignlanguage}[2]{{%
  \expandafter\ifx\csname l@#1\endcsname\relax
  \typeout{** WARNING: IEEEtran.bst: No hyphenation pattern has been}%
  \typeout{** loaded for the language `#1'. Using the pattern for}%
  \typeout{** the default language instead.}%
  \else
  \language=\csname l@#1\endcsname
  \fi
  #2}}
  \providecommand{\BIBdecl}{\relax}
  \BIBdecl
  
  \bibitem{duan2021metaverse}
  H.~Duan, J.~Li, S.~Fan, Z.~Lin, X.~Wu, and W.~Cai, ``Metaverse for social good: A university campus prototype,'' in \emph{ACM MM}, 2021, pp. 153--161.
  
  \bibitem{meng2023enabling}
  Z.~Meng, T.~Wang, Y.~Shen, B.~Wang, M.~Xu, R.~Han, H.~Liu, V.~Arun, H.~Hu, and X.~Wei, ``Enabling high quality real-time communications with adaptive frame-rate,'' in \emph{USENIX NSDI}, 2023, pp. 1429--1450.
  
  \bibitem{bhardwaj2022ekya}
  R.~Bhardwaj, Z.~Xia, G.~Ananthanarayanan, J.~Jiang, Y.~Shu, N.~Karianakis, K.~Hsieh, P.~Bahl, and I.~Stoica, ``Ekya: Continuous learning of video analytics models on edge compute servers,'' in \emph{USENIX NSDI}, 2022, pp. 119--135.
  
  \bibitem{edgeComputingMarket}
  \BIBentryALTinterwordspacing
  ``{NVIDIA unveils GPU-accelerated AI-on-5G system for edge AI, 5G and omniverse digital twins},'' {accessed 2024-01-19}. [Online]. Available: \url{{https://blogs.nvidia.com/blog/2023/02/27/mwc-ai-on-5g-system/}}
  \BIBentrySTDinterwordspacing
  
  \bibitem{edgeComputingCompany}
  \BIBentryALTinterwordspacing
  ``100 edge computing companies to watch in 2023,'' {accessed 2024-01-19}. [Online]. Available: \url{https://stlpartners.com/articles/edge-computing/edge-computing-companies-2023/}
  \BIBentrySTDinterwordspacing
  
  \bibitem{ning2020distributed}
  Z.~Ning, P.~Dong, X.~Wang, S.~Wang, X.~Hu, S.~Guo, T.~Qiu, B.~Hu, and R.~Y. Kwok, ``Distributed and dynamic service placement in pervasive edge computing networks,'' \emph{IEEE Transactions on Parallel and Distributed Systems}, vol.~32, no.~6, pp. 1277--1292, 2020.
  
  \bibitem{tourani2020democratizing}
  R.~Tourani, S.~Srikanteswara, S.~Misra, R.~Chow, L.~Yang, X.~Liu, and Y.~Zhang, ``Democratizing the edge: A pervasive edge computing framework,'' \emph{arXiv preprint arXiv:2007.00641}, 2020.
  
  \bibitem{dougherty2021apecs}
  S.~Dougherty, R.~Tourani, G.~Panwar, R.~Vishwanathan, S.~Misra, and S.~Srikanteswara, ``Apecs: A distributed access control framework for pervasive edge computing services,'' in \emph{ACM CCS}, 2021, pp. 1405--1420.
  
  \bibitem{kaiser2014efficient}
  D.~Kaiser and M.~Waldvogel, ``Efficient privacy preserving multicast dns service discovery,'' in \emph{2014 IEEE Intl Conf on High Performance Computing and Communications, 2014 IEEE 6th Intl Symp on Cyberspace Safety and Security, 2014 IEEE 11th Intl Conf on Embedded Software and Syst (HPCC, CSS, ICESS)}.\hskip 1em plus 0.5em minus 0.4em\relax IEEE, 2014, pp. 1229--1236.
  
  \bibitem{wu2016privacy}
  D.~J. Wu, A.~Taly, A.~Shankar, and D.~Boneh, ``Privacy, discovery, and authentication for the internet of things,'' in \emph{Computer Security--ESORICS 2016: 21st European Symposium on Research in Computer Security, Heraklion, Greece, September 26-30, 2016, Proceedings, Part II 21}.\hskip 1em plus 0.5em minus 0.4em\relax Springer, 2016, pp. 301--319.
  
  \bibitem{welke2016differentiating}
  P.~Welke, I.~Andone, K.~Blaszkiewicz, and A.~Markowetz, ``Differentiating smartphone users by app usage,'' in \emph{ACM UbiComp}, 2016, pp. 519--523.
  
  \bibitem{weiss2018survey}
  M.~Weiss, M.~Luck, R.~Girgis, C.~Pal, and J.~P. Cohen, ``A survey of mobile computing for the visually impaired,'' \emph{arXiv preprint arXiv:1811.10120}, 2018.
  
  \bibitem{zhu2004prudentexposure}
  F.~Zhu, M.~Mutka, and L.~Ni, ``Prudentexposure: A private and user-centric service discovery protocol,'' in \emph{Second IEEE Annual Conference on Pervasive Computing and Communications, 2004. Proceedings of the}.\hskip 1em plus 0.5em minus 0.4em\relax IEEE, 2004, pp. 329--338.
  
  \bibitem{zhou2022aadec}
  X.~Zhou, D.~He, J.~Ning, M.~Luo, and X.~Huang, ``Aadec: Anonymous and auditable distributed access control for edge computing services,'' \emph{IEEE Transactions on Information Forensics and Security}, vol.~18, pp. 290--303, 2022.
  
  \bibitem{xue2018combining}
  K.~Xue, W.~Chen, W.~Li, J.~Hong, and P.~Hong, ``Combining data owner-side and cloud-side access control for encrypted cloud storage,'' \emph{IEEE Transactions on Information Forensics and Security}, vol.~13, no.~8, pp. 2062--2074, 2018.
  
  \bibitem{zheng2022towards}
  T.~Zheng, Y.~Luo, T.~Zhou, and Z.~Cai, ``Towards differential access control and privacy-preserving for secure media data sharing in the cloud,'' \emph{Computers \& Security}, vol. 113, p. 102553, 2022.
  
  \bibitem{hu2020ghostor}
  Y.~Hu, S.~Kumar, and R.~A. Popa, ``Ghostor: Toward a secure data-sharing system from decentralized trust.'' in \emph{USENIX NSDI}, 2020, pp. 851--877.
  
  \bibitem{li2022dsos}
  H.~Li, J.~Yu, J.~Fan, and Y.~Pi, ``Dsos: A distributed secure outsourcing system for edge computing service in iot,'' \emph{IEEE Transactions on Systems, Man, and Cybernetics: Systems}, vol.~53, no.~1, pp. 238--250, 2022.
  
  \bibitem{mao2020privacy}
  Y.~Mao, W.~Hong, H.~Wang, Q.~Li, and S.~Zhong, ``Privacy-preserving computation offloading for parallel deep neural networks training,'' \emph{IEEE Transactions on Parallel and Distributed Systems}, vol.~32, no.~7, pp. 1777--1788, 2020.
  
  \bibitem{chen2013new}
  X.~Chen, J.~Li, J.~Ma, Q.~Tang, and W.~Lou, ``New algorithms for secure outsourcing of modular exponentiations,'' \emph{IEEE Transactions on Parallel and Distributed Systems}, vol.~25, no.~9, pp. 2386--2396, 2013.
  
  \bibitem{chen2019efficient}
  X.~Chen, Y.~Cai, Q.~Shi, M.~Zhao, B.~Champagne, and L.~Hanzo, ``Efficient resource allocation for relay-assisted computation offloading in mobile-edge computing,'' \emph{IEEE Internet of Things Journal}, vol.~7, no.~3, pp. 2452--2468, 2019.
  
  \bibitem{gao2021task}
  M.~Gao, R.~Shen, L.~Shi, W.~Qi, J.~Li, and Y.~Li, ``Task partitioning and offloading in dnn-task enabled mobile edge computing networks,'' \emph{IEEE Transactions on Mobile Computing}, 2021.
  
  \bibitem{gao2019dynamic}
  Y.~Gao, W.~Tang, M.~Wu, P.~Yang, and L.~Dan, ``Dynamic social-aware computation offloading for low-latency communications in iot,'' \emph{IEEE Internet of Things Journal}, vol.~6, no.~5, pp. 7864--7877, 2019.
  
  \bibitem{park2018cooperative}
  G.~S. Park and H.~Song, ``Cooperative base station caching and x2 link traffic offloading system for video streaming over sdn-enabled 5g networks,'' \emph{IEEE Transactions on Mobile Computing}, vol.~18, no.~9, pp. 2005--2019, 2018.
  
  \bibitem{zhang2025optimizing}
  L.~Zhang, M.~Wang, L.~Wang, Z.~Chen, and H.~Zhang, ``Optimizing vehicle edge computing task offloading at intersections: a fuzzy decision-making approach,'' \emph{The Journal of Supercomputing}, vol.~81, no.~1, p.~29, 2025.
  
  \bibitem{jia2024deep}
  M.~Jia, L.~Zhang, J.~Wu, Q.~Guo, G.~Zhang, and X.~Gu, ``Deep multi-agent reinforcement learning for task offloading and resource allocation in satellite edge computing,'' \emph{IEEE Internet of Things Journal}, 2024.
  
  \bibitem{chen2024multi}
  Y.~Chen, J.~Zhao, Y.~Wu, J.~Huang, and X.~Shen, ``Multi-user task offloading in uav-assisted leo satellite edge computing: A game-theoretic approach,'' \emph{IEEE Transactions on Mobile Computing}, 2024.
  
  \bibitem{wang2024ddqn}
  S.~Wang, Z.~Lu, H.~Gui, X.~He, S.~Zhao, Z.~Fan, Y.~Zhang, and S.~Pang, ``Ddqn-based online computation offloading and application caching for dynamic edge computing service management,'' \emph{Ad Hoc Networks}, p. 103681, 2024.
  
  \bibitem{sun2024joint}
  H.~Sun, Y.~Zhou, H.~Zhang, L.~Ale, H.~Dai, and N.~Zhang, ``Joint optimization of caching, computing and trajectory planning in aerial mobile edge computing networks: A maddpg approach,'' \emph{IEEE Internet of Things Journal}, 2024.
  
  \bibitem{li2025edge}
  C.~Li, X.~Deng, R.~Huang, L.~Zheng, and C.~Yang, ``Edge computing offload and resource allocation strategy with pairing theory,'' in \emph{International Conference on Mobile Multimedia Communications}.\hskip 1em plus 0.5em minus 0.4em\relax Springer, 2025, pp. 283--295.
  
  \bibitem{aloufi2021edgy}
  R.~Aloufi, H.~Haddadi, and D.~Boyle, ``Edgy: On-device paralinguistic privacy protection,'' in \emph{ACM MobiCom Workshop}, 2021, pp. 3--5.
  
  \bibitem{do2011profiling}
  A.~V. Do, J.~Chen, C.~Wang, Y.~C. Lee, A.~Y. Zomaya, and B.~B. Zhou, ``Profiling applications for virtual machine placement in clouds,'' in \emph{IEEE CLOUD}, 2011, pp. 660--667.
  
  \bibitem{fuchsbauer2020blind}
  G.~Fuchsbauer, A.~Plouviez, and Y.~Seurin, ``Blind schnorr signatures and signed elgamal encryption in the algebraic group model,'' in \emph{IACR EUROCRYPT}, 2020, pp. 63--95.
  
  \bibitem{turan2020tmps}
  M.~S. Turan, ``Tmps: Ticket-mediated password strengthening,'' in \emph{CT-RSA}, vol. 12006, 2020, p. 225.
  
  \bibitem{elgamal1985public}
  T.~ElGamal, ``A public key cryptosystem and a signature scheme based on discrete logarithms,'' \emph{IEEE Transactions on Information Theory}, vol.~31, no.~4, pp. 469--472, 1985.
  
  \bibitem{golle2004universal}
  P.~Golle, M.~Jakobsson, A.~Juels, and P.~Syverson, ``Universal re-encryption for mixnets,'' in \emph{CT-RSA}, 2004, pp. 163--178.
  
  \bibitem{green2007identity}
  M.~Green and G.~Ateniese, ``Identity-based proxy re-encryption,'' in \emph{ACNS}, 2007, pp. 288--306.
  
  \bibitem{tsiounis1998security}
  Y.~Tsiounis and M.~Yung, ``On the security of elgamal based encryption,'' in \emph{IACR PKC}, 1998, pp. 117--134.
  
  \bibitem{wang2024veriedge}
  X.~Wang, R.~Yu, D.~Yang, H.~Gu, and Z.~Li, ``Veriedge: Verifying and enforcing service level agreements for pervasive edge computing,'' in \emph{IEEE INFOCOM}, 2024, pp. 2149--2158.
  
  \bibitem{kate2023flexirand}
  A.~Kate, E.~V. Mangipudi, S.~Maradana, and P.~Mukherjee, ``Flexirand: Output private (distributed) vrfs and application to blockchains,'' in \emph{ACM CCS}, 2023, pp. 1776--1790.
  
  \bibitem{hao2023robust}
  X.~Hao, C.~Lin, W.~Dong, X.~Huang, and H.~Xiong, ``Robust and secure federated learning against hybrid attacks: A generic architecture,'' \emph{IEEE Transactions on Information Forensics and Security}, 2023.
  
  \bibitem{ben2005universal}
  M.~Ben-Or, M.~Horodecki, D.~W. Leung, D.~Mayers, and J.~Oppenheim, ``The universal composable security of quantum key distribution,'' in \emph{IACR TCC}, 2005, pp. 386--406.
  
  \bibitem{canetti2001universally}
  R.~Canetti, ``Universally composable security: A new paradigm for cryptographic protocols,'' in \emph{IEEE FOCS}, 2001, pp. 136--145.
  
  \bibitem{grpc}
  \BIBentryALTinterwordspacing
  ``{gRPC: A high performance, open-source universal RPC framework},'' {accessed 2024-01-19}. [Online]. Available: \url{{https://grpc.io/}}
  \BIBentrySTDinterwordspacing
  
  \bibitem{crypto}
  \BIBentryALTinterwordspacing
  ``{PyCrypto: The Python cryptography toolkit},'' {accessed 2024-01-19}. [Online]. Available: \url{https://github.com/pycrypto/pycrypto}
  \BIBentrySTDinterwordspacing
  
  \bibitem{charm_crypto}
  \BIBentryALTinterwordspacing
  ``{Charm: A framework for rapidly prototyping cryptosystems},'' {accessed 2024-01-19}. [Online]. Available: \url{{https://github.com/JHUISI/charm}}
  \BIBentrySTDinterwordspacing
  
  \bibitem{elgamal_implement}
  \BIBentryALTinterwordspacing
  ``Python implementation of the elgamal crypto system,'' {accessed 2024-01-19}. [Online]. Available: \url{https://github.com/RyanRiddle/elgamal}
  \BIBentrySTDinterwordspacing
  
  \bibitem{farhadi2021service}
  V.~Farhadi, F.~Mehmeti, T.~He, T.~F. La~Porta, H.~Khamfroush, S.~Wang, K.~S. Chan, and K.~Poularakis, ``Service placement and request scheduling for data-intensive applications in edge clouds,'' \emph{IEEE/ACM Transactions on Networking}, vol.~29, no.~2, pp. 779--792, 2021.
  
  \bibitem{aral2018dependency}
  A.~Aral and I.~Brandic, ``Dependency mining for service resilience at the edge,'' in \emph{IEEE/ACM SEC}, 2018, pp. 228--242.
  
  \end{thebibliography}




\vspace{-11mm}
\begin{IEEEbiography}[{\includegraphics[width=1in,height=1.25in,clip,keepaspectratio]{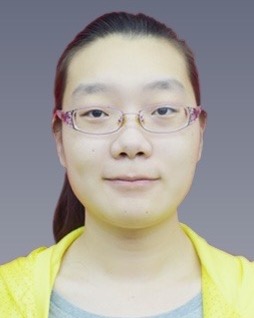}}]{Xiaojian Wang}(Student Member 2021) received her B.E. degree from Taiyuan University of Technology, China, in 2017 and received her M.S. degree in Computer Science from University of West Florida, FL, USA and Taiyuan University of Technology, China, in 2020. She is now a Ph.D. student in the department of Computer Science, College of Engineering at North Carolina State University. Her research interests include satellite network, security, blockchain, edge computing and so on. 
\end{IEEEbiography}

\vspace{-12mm}
\begin{IEEEbiography}[{\includegraphics[width=1in,height=1.25in,clip,keepaspectratio]{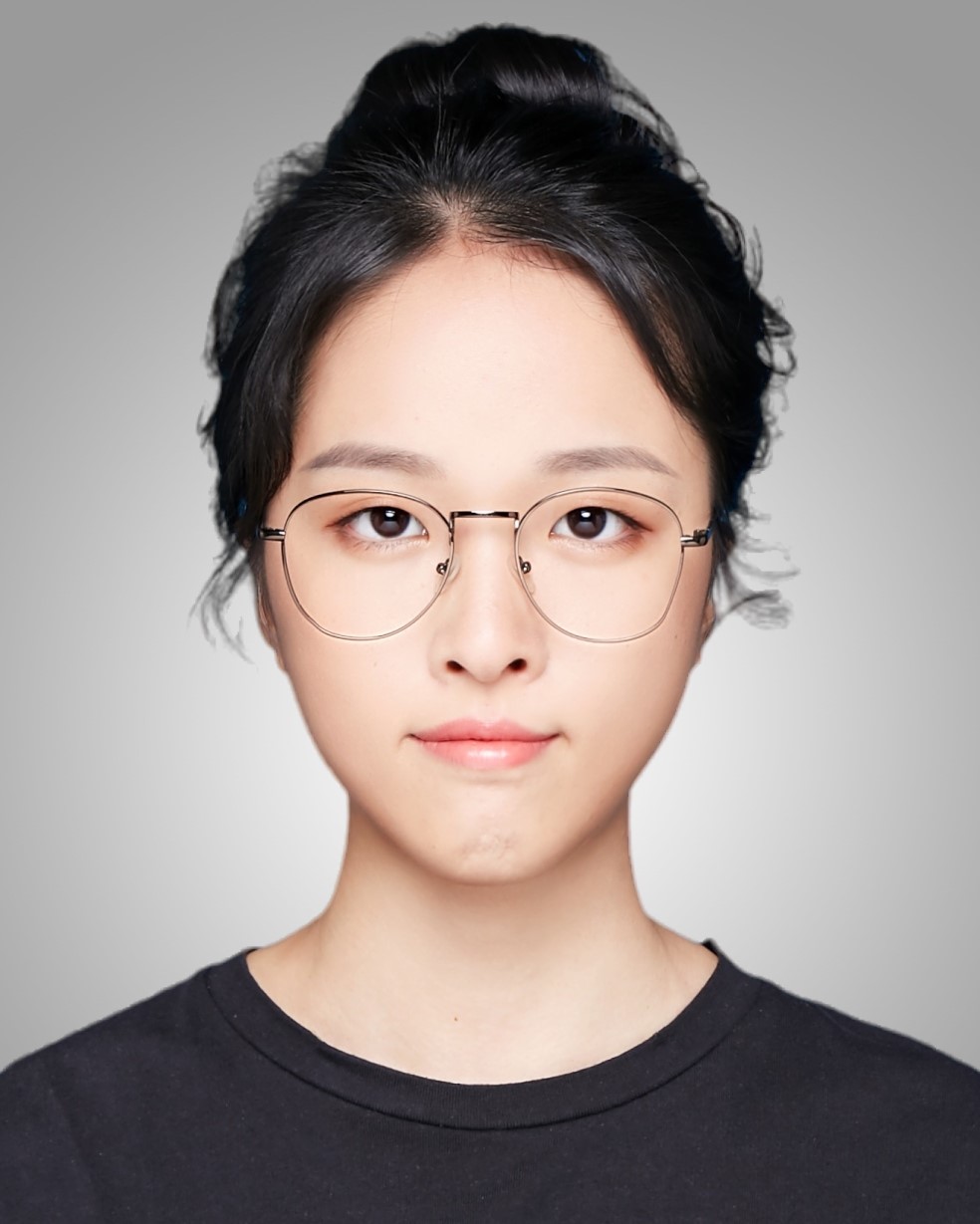}}]{Huayue Gu}(Student Member 2021) received her B.E. degree (2019) in Computer Science from Nanjing University of Posts and Telecommunications, Jiangsu, China, and M.S. degree (2021) in Computer Science from University of California, Riverside, CA, USA. Currently, she is a Ph.D. student in the Computer Science department at North Carolina State University. Her research interests are quantum networking, quantum communication, data analytics, etc.
\end{IEEEbiography}
\vspace{-12mm}
\begin{IEEEbiography}[{\includegraphics[width=1in,height=1.25in,clip,keepaspectratio]{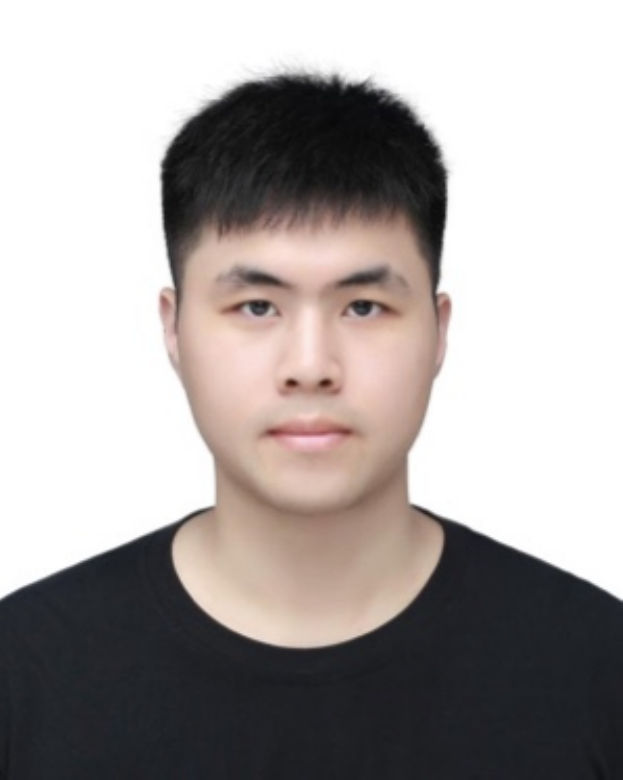}}]{Zhouyu Li}
(Student Member 2021) received his B.S. degree from Central South University, Changsha, China, in 2019 and his M.S. degree from Georgia Institute of Technology, Atlanta, U.S., in 2020. Currently, he is a Ph.D. student of Computer Science at North Carolina State University. His research interests include privacy, cloud/edge computing, network routing, etc.
\end{IEEEbiography}
\begin{IEEEbiography}[{\includegraphics[width=1in,height=1.25in,clip,keepaspectratio]{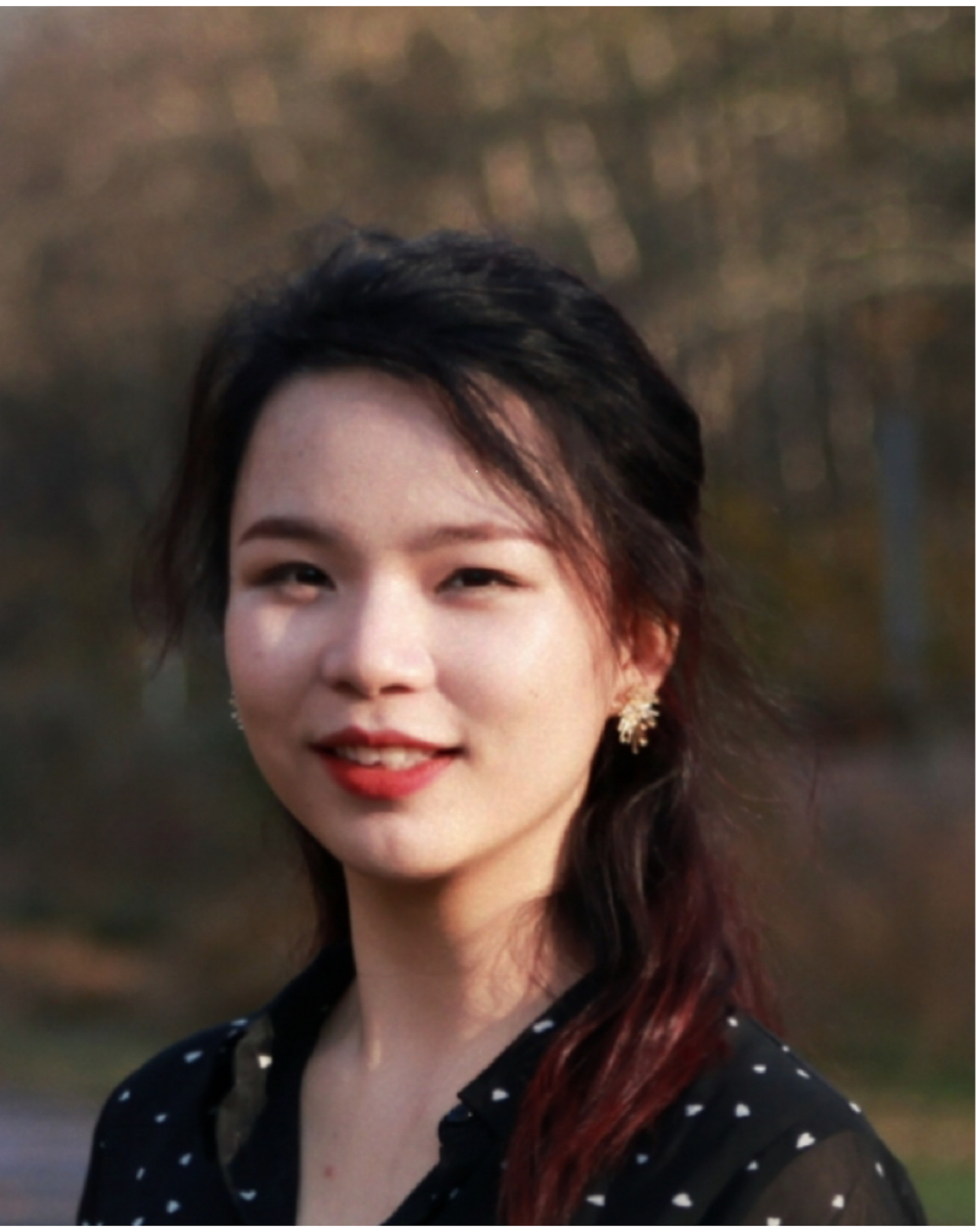}}]{Fangtong Zhou}(Student Member 2021) received her B.E. degree (2018) in Electrical Engineering and Automation from Harbin Institute of Technology, Harbin, China and M.S. degree (2020) in Electrical Engineering from Texas A\&M University, College Station, Texas, USA. Currently she is a Ph.D candidate in the School of Computer Science at North Carolina State University. Her research interests include machine learning in computer networking, like federated learning, reinforcement learning for resource provisioning, etc. 
\end{IEEEbiography}

\vspace{-12mm}
\begin{IEEEbiography}
[{\includegraphics[width=1in,height=1.25in,clip,keepaspectratio]{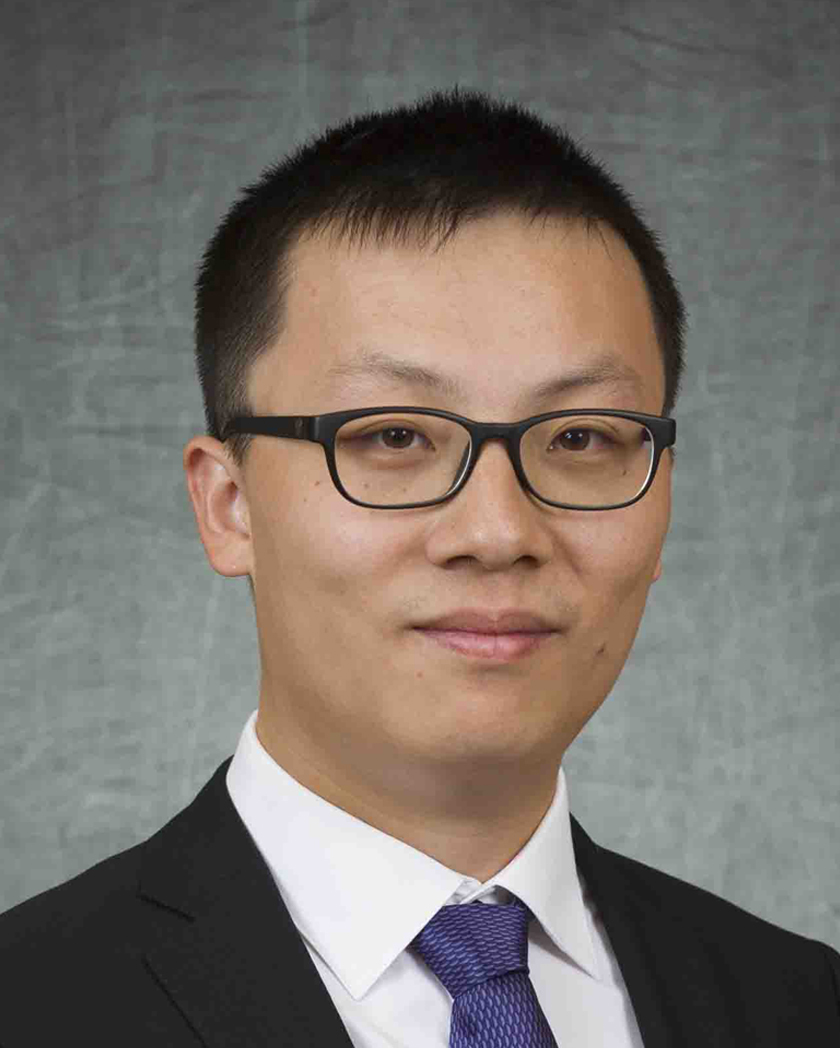}}]
{Ruozhou Yu} (Student Member 2013, Member 2019, Senior Member 2021) is an Assistant Professor of Computer Science at North Carolina State University, USA. He received his PhD degree (2019) in Computer Science from Arizona State University, USA. His research interests include internet-of-things, cloud/edge computing, smart networking, algorithms and optimization, distributed machine learning, security and privacy, blockchain, and quantum networking. He has served or is serving on the organizing committees of IEEE INFOCOM 2022-2023 and IEEE IPCCC 2020-2023, as a TPC Track Chair for IEEE ICCCN 2023, and as members of the technical committee of IEEE INFOCOM 2020-2024 and ACM Mobihoc 2023. He is a recipient of the NSF CAREER Award in 2021.
\end{IEEEbiography}

\vspace{-12mm}
\begin{IEEEbiography}
[{\includegraphics[width=1in,height=1.25in,clip,keepaspectratio]{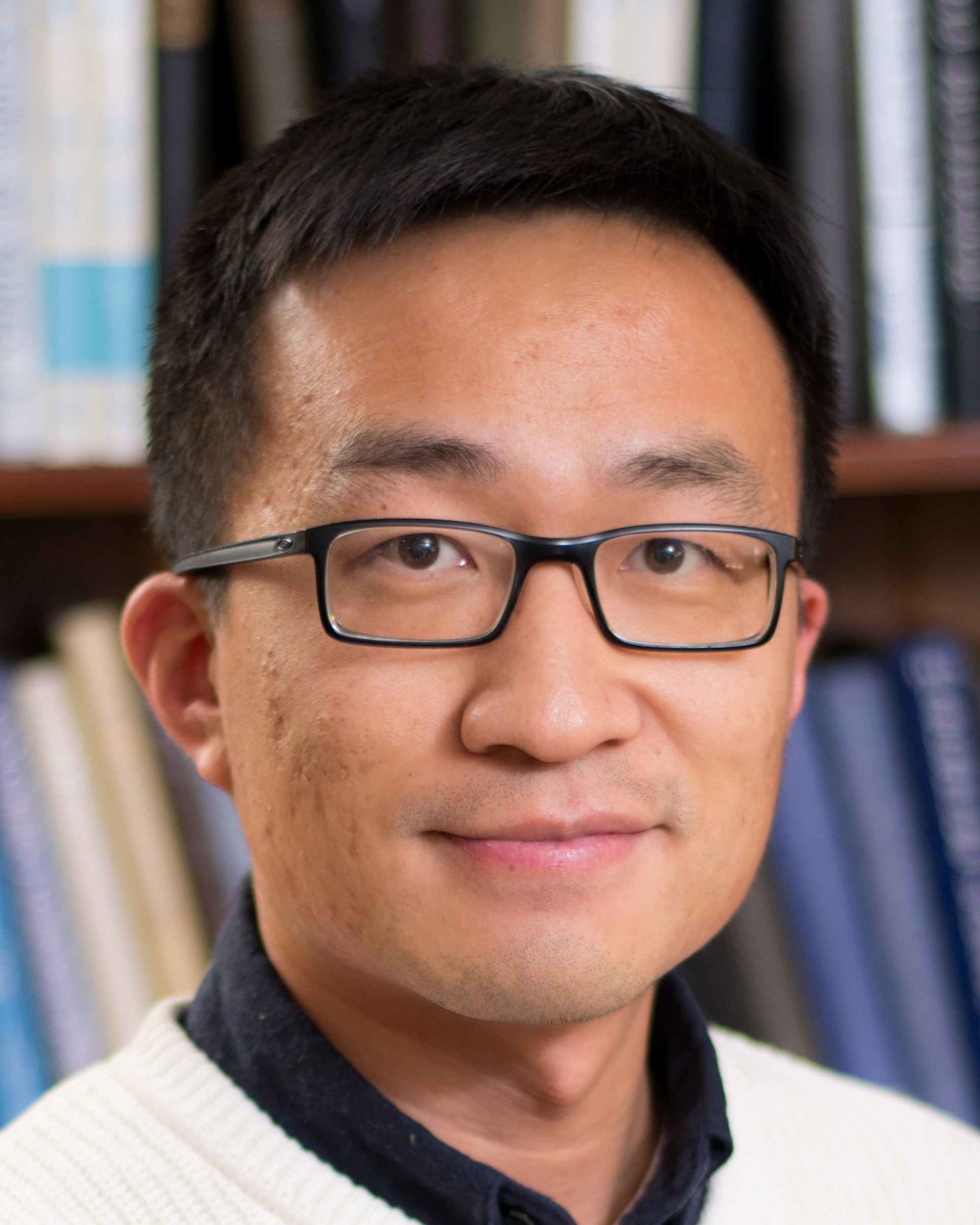}}]
{Dejun Yang} (Senior Member, IEEE) received the B.S. degree in computer science from Peking University, Beijing, China, and the Ph.D. degree in computer science from Arizona State University, Tempe, AZ, USA. \\
He is currently an Associate Professor of Computer Science with the Colorado School of Mines, Golden, CO, USA. His research interests include the Internet of Things, networking, and mobile sensing and computing with a focus on the application of game theory, optimization, algorithm design, and machine learning to resource allocation, security, and privacy problems.\\
Prof. Yang has received the IEEE Communications Society William R. Bennett Prize in 2019. He has served as the TPC Vice-Chair for Information Systems for IEEE International Conference on Computer Communications (INFOCOM). He currently serves an Associate Editor for the IEEE Transactions on Mobile Computing, IEEE Transactions on Network Science and Engineering, and IEEE Internet of Things Journal.
\end{IEEEbiography}
\vspace{-12mm}
\begin{IEEEbiography}[{\includegraphics[width=1in,height=1.25in,clip,keepaspectratio]{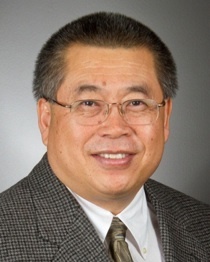}}]
{Guoliang Xue} (Member 1996, Senior Member 1999, Fellow 2011) is
a Professor of Computer Science in the School of Computing and Augmented
Intelligence at Arizona State University.
His research interests span the areas of
Internet-of-things,
cloud/edge/quantum computing and networking,
crowdsourcing and truth discovery,
QoS provisioning and network optimization,
security and privacy,
optimization and machine learning.
He received the IEEE Communications Society William R. Bennett Prize in 2019.
He is an Associate Editor of IEEE Transactions on Mobile Computing,
as well as a member of the Steering Committee of this journal.
He served on the editorial boards of
IEEE/ACM Transactions on Networking
and
IEEE Network Magazine,
as well as the Area Editor of
IEEE Transactions on Wireless Communications, overseeing 13 editors
in the Wireless Networking area.
He has served as VP-Conferences of the IEEE Communications Society.
He is the Steering Committee Chair of IEEE INFOCOM.
\end{IEEEbiography}


\end{document}